\documentclass[11pt]{article}

\usepackage{amsmath,amsfonts,amssymb,amscd,amsthm}
\usepackage{graphicx}
\usepackage{comment}
\usepackage{color}
\usepackage{enumitem}
\usepackage{booktabs}
\usepackage{hyperref}
\usepackage{bbm}
\usepackage{float}
\usepackage{subcaption}
\usepackage{algpseudocode}
\usepackage{natbib}
\usepackage{mathrsfs}

\newif\ifblog
\newif\iftex
\blogfalse
\textrue

\def\P{{\mathbb P}}
\def\E{{\mathbb E}}

\def\R{{\mathbb R}}
\def\N{{\mathbb N}}

\def\R{{\mathbb R}}

\def\Y{{\mathcal Y}}
\def\X{{\mathcal X}}

\def\M{{\mathcal{M}}}

\newcommand{\bed}{\begin{displaymath}}
\newcommand{\eed}{\end{displaymath}}
\newcommand{\bea}{\bed\begin{array}{rl}}
	\newcommand{\eea}{\end{array}\eed}

\newcommand{\barray}{\begin{array}{ll}}
	\newcommand{\earray}{\end{array}}

\providecommand{\ind}{\mathbbm{1}}

\providecommand{\half}{\ensuremath{\frac{1}{2}}}

\newcommand*{\defeq }{\mathrel{\vcenter{\baselineskip0.5ex \lineskiplimit0pt
			\hbox{\scriptsize.}\hbox{\scriptsize.}}}%
	=}

\newtheorem{theorem}{Theorem}[section]

\newtheorem{proposition}[theorem]{Proposition}

\newtheorem{remark}[theorem]{Remark}

\newcommand{\LLL}{\mathbb{L}}

\title{S-shaped Utility Maximization with VaR Constraint and Partial Information}

\author{Dongmei Zhu\footnote{School of Economics and Management, Southeast University, Nangjing, China.
zhudongmei@seu.edu.cn.
Supported  in part by the National Social Science Fund of China (Grant No.24CJY063).}, Ashley Davey\footnote{Department of Mathematics, Imperial College, London SW7 2BZ, UK. ashley.davey18@imperial.ac.uk.}, Harry Zheng\footnote{Department of Mathematics, Imperial College, London SW7 2BZ, UK. h.zheng@imperial.ac.uk. Supported  in part by EPSRC (UK)  grant (EP/V008331/1).}}

\date{}

\begin{document}
\maketitle

\begin{abstract}
We study S-shaped utility maximisation  with VaR constraint and unobservable drift coefficient. Using the Bayesian filter, the concavification principle, and the change of measure, we give a semi-closed integral representation for the dual value function and find a critical wealth level that determines  if  the constrained problem   admits a unique optimal solution and Lagrange multiplier or is infeasible.  We also propose three algorithms (Lagrange, simulation, deep neural network) to solve the problem and compare their performances with numerical examples. 
\end{abstract}




\section{Introduction} \label{sec_intro}
 \setcounter{equation}{0}
Optimal portfolio via expected utility maximization has been extensively studied, see  \cite{pham09} for expositions. The S-shaped utility has drawn particularly great attention since the ground-breaking work of  
 \cite{KT} on  the prospect theory.  \cite{car00} is the first in  solving S-shaped utility maximization with the concavification principle in a complete market. There are many papers in the literature on the subject, for example, 
\cite{bkp} incorporate prospect theory and derive closed-form solutions for optimal portfolio choice under loss aversion;
 \cite{JZ} study a general continuous-time behavioural portfolio selection model with S-shaped utility  and probability distortion;
\cite{ing13} discuss realization utility with reference-dependent preferences.

S-shaped utility maximization may lead to extreme loss due to its risk-seeking nature in the loss region. To mitigate this risk, one may incorporate some risk measures into the model.  The most common one is the value at risk (VaR), defined as the maximum portfolio loss that may occur during a given period with a pre-set confidence level, which satisfies the regulatory  requirements. There are also many papers in the literature on the subject, for example, 
 \cite{basak01} embed the VaR  into a utility maximization framework and study its implication for optimal portfolio policies. 
 \cite{yiu04} imposes the VaR as a dynamic constraint and derive the optimal constrained portfolio allocation by dynamic programming technique.
 \cite{chen18} focus on a utility maximization problem under multiple VaR-type constraints and their effects.
\cite{ben22} discuss a Merton problem with an additional variance of terminal wealth term in objective function, leading to a time-inconsistent problem.
 \cite{dong20} study   S-shaped utility maximization  with a VaR constraint and solve the  problem by  the dual control method.

Aforementioned papers assume fully observable models with deterministic coefficients or observable random coefficients. 
In real financial markets investors often can only observe partial information of risky assets, not  full information needed for valuation and optimization, for example, stock price processes but not stock growth rates,  or equity values but not  firm values, etc. 
To circumvent these issues, the filtering theory (see \cite{Bain2009}) is normally used to extract the information of unobservable random parameters with observable  information,  see \cite{det86,kar91} for introduction of this field in asset pricing and utility maximization. 
There are three typical models for unobservable random parameters, including linear diffusion models,  leading to the Kalman filter, see \cite{bren06}; finite state Markov chain models, leading to the Wonham filter, see \cite{rie05,sas07}, and random vector models with prior distribution,  leading to the Bayesian filter, see \cite{fran19, eks16}. 

In this paper we extend the work of   \cite{dong20} to models with partial information. Specifically, we assume the drift coefficient of the risky asset is an unobservable random variable with some prior distribution.  Using the Bayesian filter, we can transform the model into an equivalent fully observable one, which results in an additional filtered state process.  The dual control approach in   \cite{dong20} no longer applies as   the joint distribution of the dual and filtered state processes is unknown, which is  in sharp contrast to \cite{dong20} where   only the distribution of the dual state process is needed and is known to be lognormally distributed, so one can easily compute the dual value function or its integral representation, then find the primal 
value function and optimal control with the primal-dual relation.  With an additional filtered state process, the distribution of the dual process depends on that of the filtered state process and is in general unknown, which makes difficult to express the dual value function in semi-closed integral form. 

To  overcome the difficulty, we use a measure change to 
 reduce the dimension of the dual state variables by one when  the prior distribution of unobservable drift coefficient is a discrete distribution with two states and then characterize the dual value function with a semi-closed integral representation. 
 We find a critical wealth level that determines  if  the S-shaped utility maximization with VaR constraint and partial information  admits a unique optimal solution and Lagrange multiplier or is infeasible, that is, the VaR constraint is not satisfied for any admissible control strategies. We give a constructive proof of our main result, Theorem \ref{thm_feasible}, which leads to an exact algorithm, called Lagrange algorithm, to solve the problem numerically. We also propose two other algorithms to solve the dual problem numerically, one is 
 Monte Carlo simulation as both dual and filtered state processes can be easily simulated, albeit their joint distribution is unknown, 
 the other is   the Physics-Informed Neural Network (PINN) method (see \cite{raissi19}) that approximates the dual value function with a neural network and uses the dual HJB equation as  a loss function.   Deep learning has been used for solving HJB equations for various stochastic control problems without additional constraints, see for example (\cite{davey22, han18, wang22}).  
 We extend the scope of these papers with an input parameter representing a Lagrange multiplier to solve control problems with additional VaR constraint.

The rest of the paper is organized as follows. In Section \ref{model} we formulate the S-shaped utility maximization problem with VaR constraint and unobservable drift coefficient and discuss the Bayesian filter and the dual formulation. In Section \ref{sec_main} we state the main result of the paper, Theorem \ref{thm_feasible}, that shows there is a critical wealth level for the existence of optimal solution and characterize the optimal terminal wealth and the corresponding Lagrange multiplier. In Section \ref{sec_algo} we propose three algorithms (Lagrange, simulation, PINN) for solving the problem. In Section \ref{sec_numerics} we provide numerical examples with our algorithms. Section \ref{sec_conc} concludes. \

\section{Model and Equivalent Problem} \label{model}
\setcounter{equation}{0}
Let $(\Omega,\mathbb{F},\cal{F},\mathbb{P})$ be a filtered probability space, where $\mathbb{P}$ is the probability measure and the filtration  $\mathbb{F}=\{{\cal{F}}_t,t\in[0,T]\}$ satisfies the usual conditions.
The market consists of one riskless asset $S_0$
and one risky asset $S$, satisfying, for $0\leq t\leq T$,  
\begin{eqnarray*}
dS_0(t)&=&r S_0(t)dt,\nonumber \\
dS(t)&=&\mu S(t) dt+\sigma S(t) dW(t), 
\end{eqnarray*}
where $\{W(t),t\in[0,T]\}$ is a standard Brownian motion, adapted to the filtration $\mathbb{F}$, $r$ and $\sigma$ are positive constants, and $\mu$ is a  ${\cal F}_0$ measurable random variable. We assume $\mu$ and $W$ are unobservable and independent of each other.  
Let $\mathbb{F}^S=\{{\cal{F}}_t^S,t\in[0,T]\}$ be the natural filtration generated by the risky asset $S$, augmented with all $\mathbb{P}$-null sets in ${\cal{F}}$. The filtration $\mathbb{F}^S$ is observable and $\mathbb{F}^S\subset \mathbb{F}$. We further assume that random variable $\mu$ takes two values $\mu^h$ and $\mu^l$ with probability $p$ and $1-p$. To avoid triviality, we assume $\mu^l<\mu^h$ and $p\in (0,1)$. There is only one risky asset in the market for simplicity, which can be easily generalized to multiple risky assets with correlated Brownian motions.

Let $\pi(t)$ be  the proportion of wealth invested in the risky asset at time $t$, then the wealth process $X$ satisfies the following stochastic differential equation (SDE): 
\begin{equation}\label{dx}
	dX(t)=X(t)(r+\pi(t)(\mu-r))dt+X(t)\pi(t)\sigma dW(t),\  X(0)=x_0, \quad 0\leq t\leq T,
\end{equation}
where  $\pi$ is ${\cal{F}}^S$-progressively measurable and satisfies  $\mathbb{E}\left[\int_0^T|\pi(t)|^2dt\right]<\infty$, called an {\it admissible control}. The set of all admissible controls  is denoted by ${\cal{A}}$.
We consider a general S-shaped utility function  given by 
$$
	U(x)=\left\{
	\begin{array}{lll}
		-\infty,\ x<0,\\
		-U_2(\theta-x),\ 0\leq x<\theta, \\
	    U_1(x-\theta),\ x\geq \theta,				
	\end{array}
	\right.
$$
where $U_1, U_2$ are  strictly increasing, strictly concave, continuously differentiable  with $U_1(0)=U_2(0)=0$, and $\theta$ is a positive constant.
Additionally, $\lim_{x\rightarrow +\infty}U_1(x)=+\infty$, $\lim_{x\rightarrow +\infty}U'_1(x)=0$, $\lim_{x\rightarrow +\infty}\frac{xU'_1(x)}{U_1(x)}<1$ and  $\lim_{x\rightarrow 0^+}U'_i(x)=+\infty$, for $i=1,2$.
In what follows, $I_i$ denotes the inverse function of $U'_i$ for $i=1,2$.
$U$ is convex for $0\leq x\leq \theta$ and concave for $x\geq \theta$, indicating the behavioral change from risk seeking to risk averse at a reference point $\theta$.  
Our problem is to maximize the expected utility of  terminal wealth with a quantile constraint: 
\begin{equation}\label{utility}
	\left\{
	\begin{array}{l}
	\sup_{\pi\in {\cal{A}}}\mathbb{E}[U(X^{\pi}(T)))],\\
	{\text{s.t.}}\ X^{\pi}(t)\ {\text{satisfies}}\ (\ref{dx}),\\
	\mathbb{P}(X^{\pi}(T)\geq L)\geq 1-\varepsilon,
	\end{array}
	\right.
	\end{equation}
where $0\leq \varepsilon\leq 1$ is a constant given in advance.
In this paper, we assume $L<\theta$. The case $L\geq \theta$ can be similarly discussed.

\subsection{ Filtering and equivalent formulation} \label{sec_filter}
We now transform the primal problem (\ref{utility}) into an equivalent completely observable  problem. 
Denote the filter estimate of $\mu$ by 
$$\hat{\mu}(t)=\mathbb{E}[\mu|{\cal{F}}_t^S]$$
 and the innovation process $\hat W$ by
\[
\hat{W}(t)\defeq \sigma^{-1}\int_0^t(\mu-\hat{\mu}(s))ds+W(t),\quad  t\in[0,T].
\]
Then $\hat W$ is a $(\mathbb{P},{\mathbb{F}}^S)$-Brownian motion (see \cite{sas07}).
We can rewrite equivalently the asset price process $S$ as
\begin{equation}\label{sn}
dS(t)=\hat{\mu}(t)S(t)dt+\sigma S(t) d\hat{W}(t),
\end{equation}
and the wealth process $X$ as 
\begin{equation}\label{xn}
dX(t)=X(t)(r+\pi(t)(\hat{\mu}(t)-r))dt+X(t)\pi(t)\sigma d\hat{W}(t).
\end{equation}
The filtered drift process $\hat{\mu}$ satisfies the SDE:
\begin{equation}\label{mu}
	d\hat{\mu}(t)=\psi(\hat\mu(t))d\hat{W}(t),
	\end{equation}
where $\psi(u)=\sigma^{-1}(u-\mu^l)(\mu^h-u)$, and $\hat\mu(0)=\mathbb{E}[\mu] =p\mu^h+(1-p)\mu^l\in (\mu^l,\mu^h)$. 

We now have a fully observed control problem with state processes $X,\hat{\mu}$ and the problem (\ref{utility}) is equivalent to the following problem:
\begin{equation}\label{uo}
\left\{
\begin{array}{l}
\sup_{\pi\in {\cal{A}}}\mathbb{E}[U(X^{\pi}(T))],\\
{\text{s.t.}}\ X^{\pi}(t)\ {\text{satisfies}}\ (\ref{xn}),\\
{\hat \mu}(t) \ {\text{satisfies}}\ (\ref{mu}),\\
\mathbb{P}(X^{\pi}(T)\geq L)\geq 1-\varepsilon,
\end{array}
\right.
\end{equation}
We can solve problem (\ref{uo}) in two steps: First,  solve an unconstrained problem:
\begin{equation}\label{util}
\left\{
\begin{array}{l}
\sup_{\pi\in {\cal{A}}}\mathbb{E}[U_{\lambda}(X^{\pi}(T))],\\
{\text{s.t.}}\ X^{\pi}(t)\ {\text{satisfies}}\ (\ref{xn}),\\
{\hat \mu}(t) \ {\text{satisfies}}\ (\ref{mu}).\\
\end{array}
\right.
\end{equation}
where 
$$U_{\lambda}(x)\defeq U(x)+\lambda {\mathbb{I}}_{\{x\geq L\}}$$
and $\lambda\geq 0$ is a Lagrange multiplier to be determined. Second, 
find $\lambda^*$ such that the quantile constraint and the complementary slackness condition are satisfied:
\begin{equation}\label{lamd}
\left\{
\begin{array}{l}
\mathbb{P}(X^{\pi^*(\lambda),\lambda}(T)\geq L)\geq 1-\varepsilon,\\
\lambda(\mathbb{P}(X^{\pi^*(\lambda),\lambda}(T)\geq L)- 1+\varepsilon)=0.\\
\end{array}
\right.
\end{equation}
The relation of problems (\ref{uo}) and  (\ref{util})  is discussed in 
	\cite[Lemma 2.3]{dong20} that shows  if there exists a nonnegative constant $\lambda^*$ such that $X^{\pi^*,\lambda^*}(T)$ solves problem (\ref{util}) and satisfies condition (\ref{lamd}), then it also solves problem (\ref{uo}).

\subsection{Concavified utility and dual problem}
The utility function $U_\lambda$ in (\ref{util}) is discontinuous at $x=L$ as well as nonconcave. We can use the concavification technique (see \cite{car00}) to solve the unconstrained problem (\ref{util}) as both state processes $X$ and $\hat \mu$ are driven by the same Brownian motion $\hat W$ and we have a complete market model. 
The concave envelope of $U$ is given by  
\begin{equation}\label{uc}
U^c(x)=\left\{
\begin{array}{ll}
-\infty, &x<0,\\
c_z x-U_2(\theta), & 0\leq x<z,\\
U_1(x-\theta), &x\geq z,
\end{array}
\right.
\end{equation}
where 
$$c_x=U'_1(x-\theta),\ x>\theta$$
 and $z>\theta$ is the unique solution to the equation
\begin{equation}\label{z}
U_1(x-\theta)+U_2(\theta)-xU'_1(x-\theta)=0.
\end{equation}
For a fixed $\lambda\geq 0$, denote by $U_{\lambda}^c$ the concave envelope of $U_{\lambda}$ and $V^c_{\lambda}$ the dual function of $U_{\lambda}^c$, defined by
$$V^c_{\lambda}(y)=\sup_{x\geq 0}\{U^c_{\lambda}(x)-xy\},\quad y>0,$$
and $x^{*,\lambda}(y)$ the maximizer of $V^c_{\lambda}(y)$. 
We can characterize $U_{\lambda}^c$ and $x^{*,\lambda}(y)$ as follows:
\begin{proposition}\label{dong}
	(\cite{dong20}) 
	Let 
\begin{equation}\label{kl}
k_{\lambda}\defeq \frac{U_2(\theta)  - U_2(\theta - L) + \lambda}{L},
\end{equation} 	
and 	 $\tilde{z}\in (\theta,z)$ is the unique solution to the equation
	\begin{equation}\label{zt}
	U_1(x-\theta)+U_2(\theta-L)-(x-L)U'_1(x-\theta)=0.
	\end{equation}
	(1) If $k_{\lambda}>c_{\tilde{z}}$, then
	\begin{equation}\label{r3-1-1}
	U_{\lambda}^c(x)=\left\{
	\begin{array}{ll}
	-\infty,& x<0,\\
	k_{\lambda}x-U_2(\theta),&0\leq x<L,\\
	c_{\tilde{z}}(x-L)-U_2(\theta - L)+\lambda,& L\leq x<\tilde{z},\\
	U_1(x-\theta)+\lambda,&x\geq \tilde{z},
	\end{array}\right.
	\end{equation}
	and 
	\begin{equation}\label{r3-1-2}
	x^{*,\lambda}(y)=\left\{
	\begin{array}{ll}
	\theta+I_1(y),&y<c_{\tilde{z}},\\
	L,&c_{\tilde{z}}\leq y<k_{\lambda},\\
	0,&y\geq k_{\lambda}.
	\end{array}
	\right.
	\end{equation}
	\\
	(2) If $k_{\lambda}\leq c_{\tilde{z}}$, then 
		\begin{equation}\label{r2-1-1}
		U_{\lambda}^c(x)=\left\{
		\begin{array}{ll}
		-\infty,& x<0,\\
		c_{z_0}x-U_2(\theta),&0\leq x<\tilde{z}_0,\\
		U_1(x-\theta)+\lambda,&x\geq \tilde{z}_0,
		\end{array}\right.
		\end{equation}
		and
		\begin{equation}\label{r2-2-l}
		x^{*,\lambda}(y)=\left\{
		\begin{array}{ll}
		\theta+I_1(y),&y<c_{\tilde{z}_0},\\
		0,&y\geq c_{\tilde{z}_0},
		\end{array}
		\right.
		\end{equation}
    where $\tilde{z}_0\in [\tilde{z},z]$ is the unique solution to the equation
    \begin{equation}\label{z0}
    U_1(x-\theta)+U_2(\theta)+\lambda-xU'_1(x-\theta)=0.
    \end{equation} 
     If $\lambda=0$, then $\tilde{z}_0=z$, $\tilde{U}_{\lambda}^c(x)$ and $x^{\lambda,0}(y)$ are given by (\ref{uc}) and (\ref{r2-2-l}) respectively.           
	\end{proposition}
Figure \ref{fig_utility1}  shows an example of  utilities $U_\lambda$ and their concavified counterparts $U_\lambda^c$. 
\begin{figure}[H] 
\centering
\begin{minipage}{.31\textwidth}
\centering
\begin{subfigure}[b]{\textwidth}
\includegraphics[width=\textwidth]{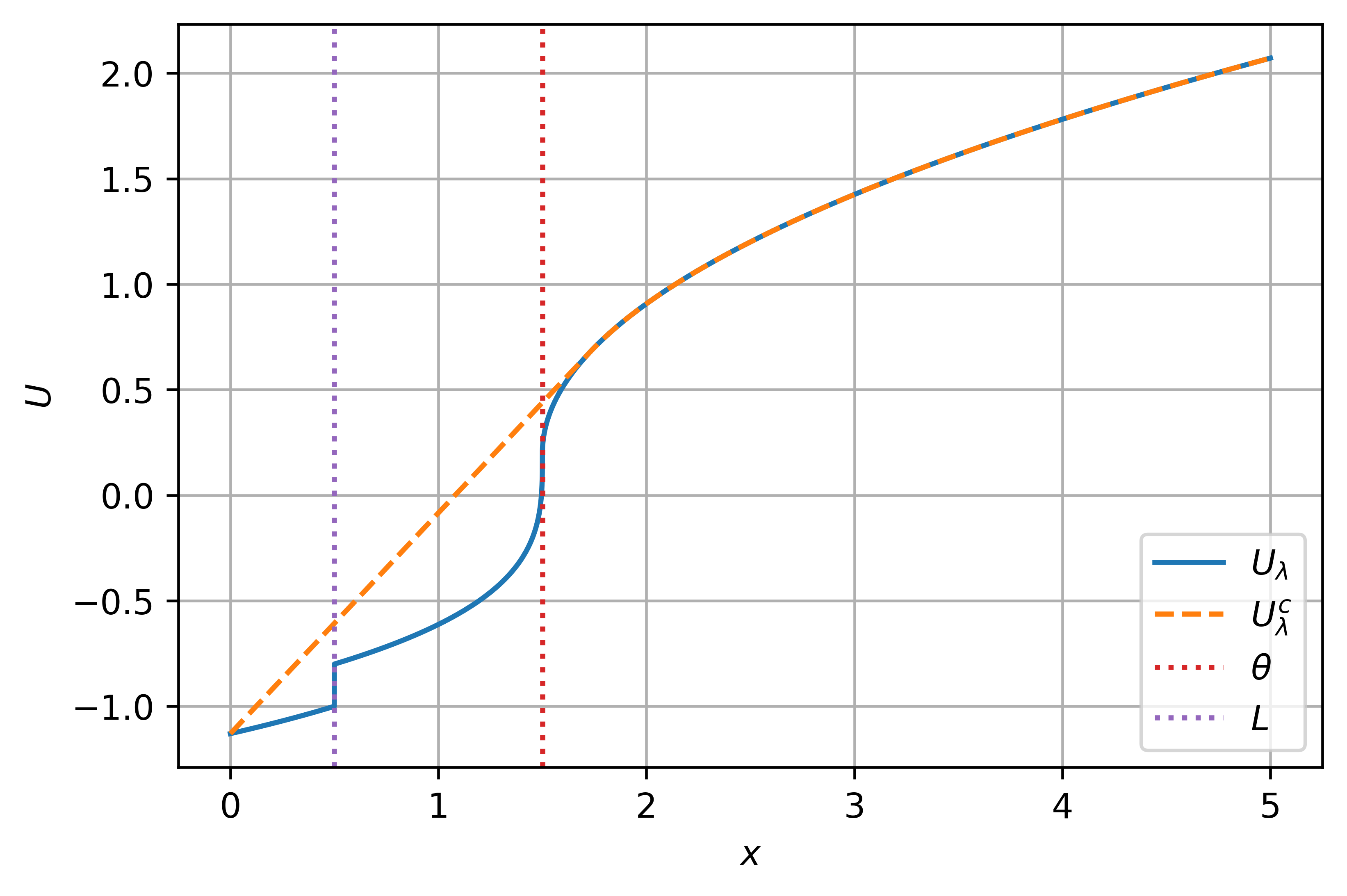}
\caption{One line segment concavified utility}
\end{subfigure}
\end{minipage}%
\;\;
\begin{minipage}{.31\textwidth}
\centering
\centering
\begin{subfigure}[b]{\textwidth}
\includegraphics[width=\textwidth]{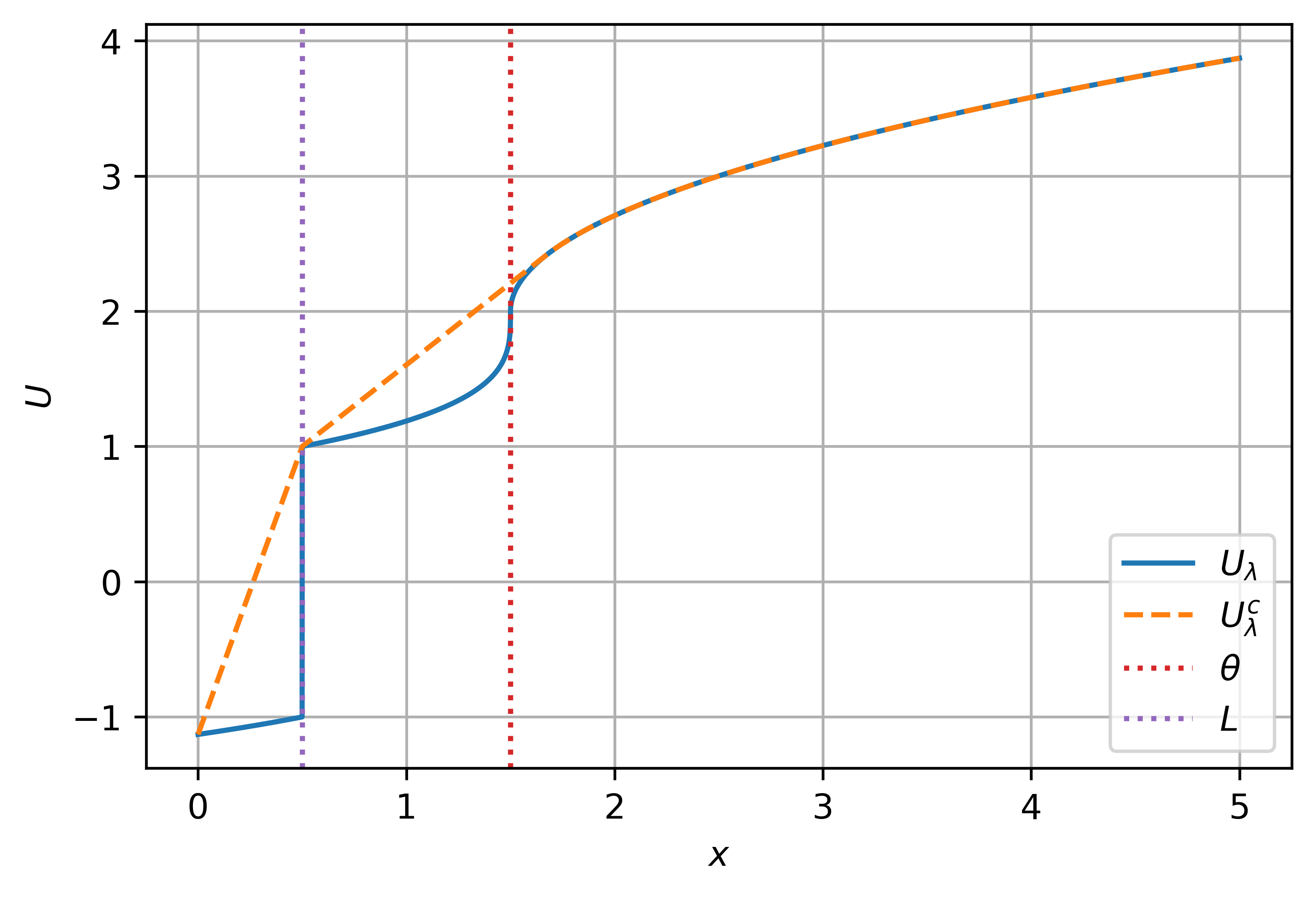}
\caption{Two line segment concavified utility}
\end{subfigure}
\end{minipage}
\caption{Utility $U_\lambda$ and concavified utility $U_\lambda^c$: (a) $k_\lambda \leq U_1'(\tilde{z} - \theta)$ and (b)  $k_\lambda > U_1'(\tilde{z} - \theta)$.}
\label{fig_utility1}
\end{figure}
Now we consider the auxiliary stochastic control problem with fixed $\lambda$: 
\begin{equation}\label{util-1}
\left\{
\begin{array}{l}
\sup_{\pi\in {\cal{A}}}\mathbb{E}[U_{\lambda}^c(X^{\pi}(T))],\\
{\text{s.t.}}\ X^{\pi}(t)\ {\text{satisfies}}\ (\ref{xn}),\\
{\hat \mu}(t) \ {\text{satisfies}}\ (\ref{mu}).\\
\end{array}
\right.
\end{equation}
Denote the value function of (\ref{util-1}) by 
\begin{align} \label{eq_value}
u_{\lambda}^c(t,x,{\hat \mu})\defeq \sup_{\pi\in {\cal{A}}}\mathbb{E}[U_{\lambda}^c(X^{\pi}(T))|X^{\pi}(t)=x,{\hat \mu}(t)={\hat \mu}]
\end{align}
and  the constraint probability function $h$ by
\begin{align} \label{eq_constraint}
h_\lambda(t, x, \hat{\mu}) := \E[\ind_{X^{\pi^*, \lambda}(T) \geq L}|X^{\pi}(t)=x,{\hat \mu}(t)={\hat \mu}],
\end{align}
where $X^{\pi^*, \lambda}$ is the optimal state process of (\ref{util-1}).
The concavification  principle states that
	 problems (\ref{util}) and (\ref{util-1}) are equivalent (see \cite[Theorem 5.1]{rei13}) and the optimal solution for (\ref{util-1}) is the same as that for (\ref{util}). 
The HJB equation for problem (\ref{util-1}) is given by
\begin{equation}\label{hjbx}
\frac{\partial u_{\lambda}^c}{\partial t}+\sup_{\pi}
\left(	(rx+x\pi(\hat{\mu}-r))\frac{\partial u_{\lambda}^c}{\partial x}+\frac{1}{2}x^2\pi^2 \sigma^2 \frac{\partial^2 u_{\lambda}^c}{\partial x^2}+\frac{1}{2}\psi^2 \frac{\partial^2 u_{\lambda}^c}{\partial \hat\mu^2}+x\pi\sigma\psi \frac{\partial^2 u_{\lambda}^c}{\partial \hat\mu \partial x}\right)=0
\end{equation} 
with the terminal condition $u_{\lambda}^c(T,x,\hat{\mu})=U_{\lambda}^c(x)$. 
This is a nonlinear PDE with two state variables and is in general difficult to solve. We may use the dual method to solve it. The dual state process $Y$ is strictly positive and satisfies the following SDE:
\begin{equation}\label{dy}
dY(t)=-Y(t)rdt - Y(t) \sigma^{-1}(\hat{\mu}(t)-r)d\hat{W}(t), \quad Y(0)=y_0.
\end{equation}
Since there is no control in SDE (\ref{dy}), the dual problem is reduced to a simple evaluation of expectation of dual function at $Y(T)$, that is, 
\begin{equation}\label{dutils}
\left\{
\begin{array}{l}
\mathbb{E}[V_{\lambda}^c(Y(T))],\\
{\text{s.t.}}\ Y(t)\ {\text{satisfies}}\ (\ref{dy}),\\
{\hat \mu}(t) \ {\text{satisfies}}\ (\ref{mu}).\\
\end{array}
\right.
\end{equation}
The dual value function is defined by
\begin{align} \label{eq_dual_value}
v_{\lambda}^c(t,y,\hat{\mu}):=\mathbb{E}[V_{\lambda}^c(Y(T))|Y(t)=y,{\hat \mu}(t)=\hat{\mu}]
\end{align}
and  the dual constraint function by 
\begin{align} \label{eq_dual_con}
g_\lambda(t, y, \hat{\mu}) := \E\left[ \ind_{x^{*,\lambda}(Y(T)) \geq L} \middle| Y(t) = y, \hat{\mu}(t) = \hat{\mu} \right],  
\end{align}
where  $\ind_S$ is an indicator  that equals 1 if $S$ happens and 0 otherwise. 
By Feynman-Kac formula, we have 
\begin{equation}\label{dhjbx}
\frac{\partial v_{\lambda}^c}{\partial t}-ry\frac{\partial v_{\lambda}^c}{\partial y}+\frac{1}{2}y^2\sigma^{-2}(\hat{\mu}-r)^2\frac{\partial^2 v_{\lambda}^c}{\partial y^2}+\frac{1}{2}\psi^2\frac{\partial^2 v_{\lambda}^c}{\partial \hat{\mu}^2}-y\sigma^{-1}(\hat{\mu}-r)\psi \frac{\partial^2 v_{\lambda}^c}{\partial \mu \partial y}=0
\end{equation} 
with the terminal condition $v_{\lambda}^c(T,y,\hat{\mu})=V_{\lambda}^c(y)$.

The optimal terminal wealth $X^{*,\lambda}(T)$ for primal problem (\ref{util}) is given by Proposition (\ref{dong}), that is, 
	$X^{*, \lambda}(T)  = x^{*, \lambda}(Y(T))$ and $y_0$ is determined by the binding budget constraint $\mathbb{E}[X^{*,\lambda}(T) Y(T)]=x_0y_0$. The optimal wealth process $X^{*,\lambda}$ can be determined by $X^{*,\lambda}(t)=\mathbb{E}[X^{*,\lambda}(T)Y(T)/Y(t)|{{\cal{F}}_t^S}]$ for $t\in[0,T]$ and the optimal control process $\pi^{*,\lambda}$ is determined by the martingale representation theorem.

\section{Main Results} \label{sec_main}
\setcounter{equation}{0}
To express  $\mathbb{E}[V_{\lambda}^c(Y(T))]$ explicitly in terms of integral representation, we need to know
the joint distribution of $Y$ and ${\hat \mu}$, which is unknown. To address this,   we employ a measure change technique (\cite{xing25}). We introduce a new process $W^Q$ by 
\[
dW^Q(t)\defeq \frac{\hat{\mu}(t) - \mu^l}{\sigma}dt+d\hat{W}(t),
\] 
and a new probability measure by 
\[
\left.\frac{dQ}{dP}\right|_{{\cal{F}}_T^S}\defeq \exp \left(-\frac{1}{2}\int_0^T\left(\frac{\hat{\mu}(u) - \mu^l}{\sigma}\right)^2du-\int_0^T\frac{\hat{\mu}(u) - \mu^l}{\sigma}d\hat{W}(u)\right).
\]
By Girsanov's theorem, $W^Q$ is a standard Brownian motion under the new probability measure $Q$.
Under this new measure $Q$, the wealth process $X$ satisfies
\[
dX(t)=X(t)(r+\theta_l \sigma\pi(t))dt+X(t)\sigma \pi(t)dW^Q(t),
\]
where $\theta_l = \frac{\mu^l - r}{\sigma}$,
and the corresponding dual process 
\[
d{\cal{Y}}(t)={\cal{Y}}(t)\left(-rdt-\theta_ldW^Q(t)\right).    
\]
Let 
$$\Phi(t)\defeq \frac{\hat{\mu}(t) - \mu^l}{\mu^h - \hat{\mu}(t)}, \quad t \in [0,T].$$
 Note that $\Phi(t)$ is well defined as $\hat{\mu}(t)\in (\mu^l, \mu^h)$ a.s. for all $t\in [0,T]$ due to $\hat\mu(0) \in (\mu^l, \mu^h)$  (see \cite[Lemma 3.1]{Decamps2005}).   Then $\Phi(t)$ satisfies the SDE:
\[
d\Phi(t)=\Theta \Phi(t)dW^Q(t)
\]
with $\Phi(0)=\phi\defeq \frac{ \hat{\mu}(0) - \mu^l}{\mu^h -  \hat{\mu}(0)}$, where $\Theta=\frac{\mu^h-\mu^l}{\sigma}$. In addition, let 
\begin{equation} \label{F(t)}
F(t)\defeq \frac{1+\Phi(t)}{1+\phi},
\end{equation}
 then $F(t)$ satisfies the SDE:
$$dF(t)=\sigma^{-1}(\hat{\mu}(t) - \mu^l) F(t)dW^Q(t)$$
with $F(0)=1$, and
\begin{eqnarray*}
	\left.F(t)=\frac{dP}{dQ}\right|_{{\cal{F}}_t^S}.
\end{eqnarray*}
Applying Ito's lemma, we deduce that
\[
d(F(t)Y(t))=-rF(t)Y(t)dt-F(t)Y(t)\theta_ldW^Q(t)
\]
with $F(0)Y(0)=y_0$, 
which yields that ${\cal{Y}}(t)=F(t)Y(t)$ under measure $Q$.
We also get that 
\begin{equation}\label{phi}
\Phi(t)=\phi \exp\left\{\Theta W^Q(t)-\frac{1}{2}\Theta^2t\right\},
\end{equation}
and
\begin{equation}\label{cy}
{\cal{Y}}(t)=y_0\exp\left\{-\theta_l W^Q(t)-\left(r+\frac{1}{2}\theta_l^2\right)t
\right\}.
\end{equation}
Using the above observations, the  dual value function in (\ref{dutils}) becomes
$$ \mathbb{E}[V_{\lambda}^c(Y(T))]
=\mathbb{E}^Q[F(T)V_{\lambda}^c({\cal{Y}}(T)/F(T))]
=\mathbb{E}^Q\left[\frac{1+\Phi(T)}{1+\phi}V_{\lambda}^c\left({\frac{1+\phi}{1+\Phi(T)}\cal{Y}}(T)\right)\right]. 
$$
Since $\Phi$ and $\cal{Y}$ can be expressed in terms of $Q$-Brownian motion $W^Q$, we can write out 
the integral representation of the dual value function (\ref{eq_dual_value})  and the constraint function (\ref{eq_dual_con}) explicitly, that is, 
\begin{align}\label{vlaue1}
v^c_{\lambda}(t,y,\hat{\mu})&=\int_{\mathbb{R}}\Psi\left(t, x, \hat{\mu}\right) V_\lambda^c\left(\frac{y \exp\left\{-\theta_lx-\left(r+\frac{1}{2}\theta_l^2\right)(T-t)\right\}}{\Psi(t, x, \hat{\mu})}\right)p(t, x)dx,\\
g_{\lambda}(t,y,\hat{\mu})&=\int_{\mathbb{R}}\Psi\left(t, x, \hat{\mu}\right) \ind_{x^{*,\lambda}\left(\frac{y \exp\left\{-\theta_lx-\left(r+\frac{1}{2}\theta_l^2\right)(T-t)\right\}}{\Psi(t, x, \hat{\mu})}\right)\geq L}p(t, x)dx, \label{eq_dual_con_exp}
\end{align}
where 
\begin{align*}
\Psi(t, x, \hat{\mu}) & \defeq \frac{1+\phi \exp\left\{\Theta x-\frac{1}{2}\Theta^2 (T-t)\right\}}{1+\phi}, \\
p(t, x) & \defeq \frac{1}{\sqrt{2\pi(T-t)}}\exp\left\{-\frac{x^2}{2(T-t)}\right\}.
\end{align*}


Condition (\ref{lamd}) can be written as
\begin{equation}\label{lamdq}
\left\{
\begin{array}{l}
\mathbb{E}^Q[F(T)\mathbb{I}_{\{X^{\pi^*,\lambda}(T)\geq L\}}]\geq 1-\varepsilon,\\
\lambda(\mathbb{E}^Q[F(T)\mathbb{I}_{\{X^{\pi^*,\lambda}(T)\geq L\}}]- 1+\varepsilon)=0.
\end{array}
\right.
\end{equation}
For a given $\lambda\geq 0$, the optimal terminal wealth is given by 
\begin{equation}\label{ox}
X^{\pi^*,\lambda}(T)=x^{*,\lambda}(Y(T)),
\end{equation}
where $x^{*,\lambda}$ is given by (\ref{r3-1-2})  or (\ref{r2-2-l}),   depending on the value $\lambda$, 
$$Y(T)={{\cal Y}(T)\over F(T)}= y_0(1+\phi)H(T),$$
 and
\begin{equation} \label{H(T)}
H(T)\defeq \frac{\exp\{-\theta_l W^Q(T)-(r+\frac{1}{2}\theta_l^2)T\}}{1+\phi \exp\left\{\Theta W^Q(T)-\frac{1}{2}\Theta^2T\right\} }.
\end{equation}
Additionally, $y_0$ is determined by the binding budget constraint $\mathbb{E}[X^{*,\lambda}(T) Y(T)]=x_0y_0$, that is,
\begin{equation}\label{bc}
\mathbb{E}^Q[F(T)X^{\pi^*,\lambda}(T)(1+\phi)H(T)]=x_0.
\end{equation}
Combining  (\ref{lamdq}) and (\ref{bc}), we can derive solutions $(y_0,\lambda^*)$.
We next state the main theorem on the existence and uniqueness of the Lagrange multiplier and the feasibility condition. 
\begin{theorem} \label{thm_feasible}
Let $x_0>0$ and $H^*_\varepsilon$ be the solution of the equation \begin{align}\label{eq_defH*}
\mathbb{E}^Q[F(T)\mathbb{I}_{\{H(T)\leq H^*_\varepsilon\}}]
=1-\varepsilon,
\end{align} where $F(T)$ and $H(T)$ are given in (\ref{F(t)}) and (\ref{H(T)}). Denote by 
\begin{equation} \label{x_epsilon}
\hat x_\varepsilon:=\mathbb{E}^Q[F(T)L\mathbb{I}_{\{H(T)<H^*_\varepsilon\}}(1+\phi)H(T)].
\end{equation}
   Then the following results hold.
\begin{enumerate}
\item If 	$x_0>\hat x_\varepsilon$, then there exists a unique $\lambda^*\geq 0$ such that $X^{\pi^*,\lambda^*}(T)$ in (\ref{ox})  is the optimal solution to problem (\ref{uo}).
\item If $x_0 = \hat x_\varepsilon$, then there is only one solution
$X^{\pi^*,\lambda^*}(T)=L\mathbb{I}_{H(T)<H^*_\varepsilon}$ a.s..
\item If $x_0< \hat x_\varepsilon$, then  
 problem (\ref{uo}) is infeasible, that is, condition  (\ref{lamd}) is not satisfied.
 \end{enumerate}
\end{theorem}

\begin{proof} We first discuss the case  $x_0>\hat x_\varepsilon$. 
		For a fixed $\lambda\geq 0$, the optimal terminal wealth $X^{\pi^*,\lambda}(T)$ is given by (\ref{ox}).
			If we can find a solution $(y_0,\lambda^*)$ to equations  (\ref{lamdq}) and (\ref{bc}), then $X^{\pi^*,\lambda^*}(T)$ is the optimal solution to problem (\ref{uo}). 
			
			 Case I: $H_{\varepsilon}^*\leq \frac{c_z}{y_0(1+\phi)}$. If we choose $\lambda=0$, then  	
			 $X^{\pi^*,\lambda}(T)$ is given by 
			 \begin{equation}\label{0x}
				X^{\pi^*, \lambda}(T)=(\theta+I_1(y_0(1+\phi)H(T)))\mathbb{I}_{\{H(T)\leq\frac{c_z}{y_0(1+\phi)}\}}.
			\end{equation}
			 We have 
\begin{eqnarray*}\mathbb{E}^Q[F(T)\mathbb{I}_{\{X^{\pi^*,\lambda}(T)\geq L\}}]
&=&\mathbb{E}^Q[F(T)\mathbb{I}_{\{H(T)\leq \frac{c_z}{y_0(1+\phi)}\}}]
\geq\mathbb{E}^Q[F(T)\mathbb{I}_{\{H(T)\leq H_{\varepsilon}^*\}}]
=1-\varepsilon.
\end{eqnarray*}
			The quantile constraint  (\ref{lamdq}) is  satisfied with  $\lambda^*=0$.
			We next show that there is a unique solution  $y_0$ to equation  (\ref{bc}).
			Denote by $f(y_0)\defeq \mathbb{E}^Q[F(T)X^{\pi^*,\lambda}(T)(1+\phi)H(T)]$. We can check that $f$ is  continuous, strictly decreasing,    and  $\lim_{y_0\rightarrow 0^+}f(y_0)=\infty, \lim_{y_0\rightarrow \infty}f(y_0)=0<x_0$, then  there exists a unique $y_0$ satisfying (\ref{bc}).
						
						 Case II: $\frac{c_{z}}{y_0(1+\phi)}<H_{\varepsilon}^*\leq \frac{c_{\tilde{z}}}{y_0(1+\phi)}$. 			In this case, if $\lambda^*=0$, then 
						 $$\mathbb{E}^Q[F(T)\mathbb{I}_{\{X^{\pi^*,\lambda}(T)\geq L\}}]=\mathbb{E}^Q[F(T)\mathbb{I}_{\{H(T)\leq \frac{c_z}{y_0(1+\phi)}\}}]
			<\mathbb{E}^Q[F(T)\mathbb{I}_{\{H(T)\leq H_{\varepsilon}^*\}}]=1-\varepsilon,$$
			which implies $\lambda^*=0$ is impossible as (\ref{lamdq}) is not satisfied. We must have  
			 $\lambda^*>0$ and the quantile constraint is binding, that is,  
			\begin{equation}\label{x-0-A}
			\mathbb{E}^Q[F(T)\mathbb{I}_{\{X^{\pi^*,\lambda}(T)\geq L\}}]=1-\varepsilon.
			\end{equation}
			Since
			$X^{\pi^*,\lambda}(T)=x^{*,\lambda}(Y(T))$ and $x^{*,\lambda}$ is given by (\ref{r3-1-2})  or (\ref{r2-2-l}), we next discuss its form.
			If $k_{\lambda}>c_{\tilde{z}}$, then
			\begin{equation}\label{x-4}X^{\pi^*,\lambda}(T)=(\theta+I_1(y_0(1+\phi)H(T)))\mathbb{I}_{\{H(T)<\frac{c_{\tilde{z}}}{y_0(1+\phi)}\}}+L\mathbb{I}_{\{\frac{c_{\tilde{z}}}{y_0(1+\phi)}\leq H(T)<\frac{k_{\lambda}}{y_0(1+\phi)}\}},
			\end{equation} and
		$$
			\mathbb{E}^Q[F(T)\mathbb{I}_{\{X^{\pi^*,\lambda}(T)\geq L\}}]=\mathbb{E}^Q[F(T)\mathbb{I}_{\{H(T)<\frac{k_{\lambda}}{y_0(1+\phi)}\}}]			>\mathbb{E}^Q[F(T)\mathbb{I}_{\{H(T)<H_\varepsilon^*\}}]
			=1-\varepsilon
			$$
				 as $k_\lambda>c_{\tilde{z}}\geq y_0(1+\phi)H_\varepsilon^*$, which is a contradiction to  (\ref{x-0-A}).
				 We must have $k_{\lambda}\leq c_{\tilde{z}}$, then 
						\begin{equation}\label{x-3}
				X^{\pi^*,\lambda}(T)=(\theta+I_1(y_0(1+\phi)H(T)))\mathbb{I}_{\{H(T)<\frac{c_{\tilde{z}_0}}{y_0(1+\phi)}\}}.
			\end{equation}
			To ensure (\ref{x-0-A}) holds, we have $c_{\tilde{z}_0}=y_0(1+\phi)H_{\varepsilon}^*$ and from (\ref{z0}) we define	 
			\begin{equation}\label{l-1-A}
			\lambda_1^*(y_0)\defeq  \tilde{z}_0U'_1(\tilde{z}_0-\theta)-U_1(\tilde{z}_0-\theta)-U_2(\theta).
			\end{equation}
To show $\lambda_1^*(y_0)>0$, define  $g(x):=xU'_1(x-\theta)-U_1(x-\theta)-U_2(\theta)$. Since $g$ is strictly decreasing and  $\tilde z_0<z$, we have  $\lambda_1^*(y_0)=g(\tilde{z}_0)>g(z)=0$. 
			The proof of existence and uniqueness of $y_0$ is similar to that in case I. 
			
		 Case III: $H_{\varepsilon}^*>\frac{c_{\tilde{z}}}{y_0(1+\phi)}$.
			In this case, similar to Case II, we have $\lambda^*>0$ and $X^{\pi^*,\lambda}(T)$ is given by (\ref{x-4}) or (\ref{x-3}). 
			If $X^{\pi^*,\lambda}(T)$ were given by (\ref{x-3}), we would have 
			$$\mathbb{E}^Q[F(T)\mathbb{I}_{\{X^{\pi^*,\lambda}(T)\geq L\}}]=\mathbb{E}^Q[F(T)\mathbb{I}_{\{H(T)<\frac{c_{\tilde{z}_0}}{y_0(1+\phi)}\}}]<\mathbb{E}^Q[F(T)\mathbb{I}_{\{H(T)<H_\varepsilon^*}]=1-\varepsilon$$ as $c_{\tilde{z}_0}\leq c_{\tilde{z}}<H_\varepsilon^*y_0(1+\phi)$, which is a contradiction to (\ref{x-0-A}).
			We must have $k_{\lambda}>c_{\tilde{z}}$ and $X^{\pi^*,\lambda}(T)$ is given by (\ref{x-4}).
						To ensure (\ref{x-0-A}) holds, we have $k_{\lambda}=y_0(1+\phi)H_{\varepsilon}^*$ and by the expression of $k_\lambda$ we define
			\begin{equation}\label{l-2-A}
			\lambda_2^*(y_0)=k_{\lambda}L+U_2(\theta-L)-U_2(\theta).
			\end{equation}
			Then 
			$$\lambda_2^*(y_0)>LU'_1(\tilde{z}-\theta)+U_2(\theta-L)-U_2(\theta)=\tilde{z}U'_1(\tilde{z}-\theta)-U_1(\tilde{z}-\theta)-U_2(\theta)=g(\tilde{z})>g(z)=0.$$
			The second equation holds as $U_1(\tilde{z}-\theta)+U_2(\theta-L)-(\tilde{z}-L)U'_1(\tilde{z}-\theta)=0$.
			
			To show the existance and uniqueness of $y_0$, define $f(y_0)\defeq \mathbb{E}^Q[F(T)X^{\pi^*,\lambda^*}(T)(1+\phi)H(T)]$. We can check that $f$  is continuous,  strictly decreasing, $\lim_{y_0\rightarrow 0^+}f(y_0)=\infty$ and $\lim_{y_0\rightarrow \infty}f(y_0)=\mathbb{E}^Q[F(T)L\mathbb{I}_{\{H(T)<H_{\varepsilon}^*\}}(1+\phi)H(T)]
			=\hat x_\varepsilon<x_0$, which shows there exists a unique $y_0$ to equation (\ref{bc}).

	We next discuss the case  $x_0=\hat x_\varepsilon$.  Combining (\ref{bc}), we have  
			\begin{equation}\label{eqL}
			\mathbb{E}^Q[F(T)L\mathbb{I}_{\{H(T)<H_{\varepsilon}^*\}}H(T)]=\mathbb{E}^Q[F(T)X^{\pi^*,\lambda^*}(T)H(T)].
			\end{equation}
				If $H_{\varepsilon}^*\leq \frac{c_z}{y_0(1+\phi)}$, then $\lambda^*=0$ and $X^{\pi^*}(T)$ is given by (\ref{0x}).  					We have 
	$$
			\mathbb{E}^Q[F(T)X^{\pi^*,\lambda^*}(T)H(T)]\geq \mathbb{E}^Q[F(T)z\mathbb{I}_{H(T)<\frac{c_z}{y_0(1+\phi)}}H(T)]			>\mathbb{E}^Q[F(T)L\mathbb{I}_{\{H(T)<H_{\varepsilon}^*\}}H(T)],
			$$
			 which is a contradiction to (\ref{eqL}). 
			Similarly, 	if $\frac{c_{z}}{y_0(1+\phi)}<H_{\varepsilon}^*\leq \frac{c_{\tilde{z}}}{y_0(1+\phi)}$, then 
			$X^{\pi^*,\lambda}(T)$ is given by (\ref{x-3}) with $c_{\tilde{z}_0}=y_0(1+\phi)H_{\varepsilon}^*$, 
			 this gives that 
			$$\mathbb{E}^Q[F(T)X^{\pi^*,\lambda^*}(T)H(T)]\geq \mathbb{E}^Q[F(T)\tilde{z}_0\mathbb{I}_{\{H(T)<H_\varepsilon^*\}}H(T)]
			>\mathbb{E}^Q[F(T)L\mathbb{I}_{\{H(T)<H_{\varepsilon}^*\}}H(T)],
			$$
						 which is again a contradiction to  (\ref{eqL}). 
			If $H_{\varepsilon}^*>\frac{c_{\tilde{z}}}{y_0(1+\phi)}$, then $X^{\pi^*,\lambda}(T)$ is given by (\ref{x-4}) with $k_{\lambda}=y_0(1+\phi)H_{\varepsilon}^*$, 
						this gives  that 
			$$\mathbb{E}^Q[F(T)X^{\pi^*,\lambda^*}(T)H(T)]\geq \mathbb{E}^Q[F(T)L\mathbb{I}_{\{H(T)<H_{\varepsilon}^*\}}H(T)].$$ 
			To ensure (\ref{eqL}) holds, we must have  $y_0=\infty$, then $k_\lambda=\infty$,  $\lambda=\infty$, and   $X^{\pi^*,\lambda^*}(T)=L\mathbb{I}_{H(T)<H_{\varepsilon}^*}$.
			
		We finally discuss the case  $x_0<\hat x_\varepsilon$. Suppose solutions exist. 
			If $H_{\varepsilon}^*\leq \frac{c_z}{y_0(1+\phi)}$, then $\lambda^*=0$ and $X^{\pi^*}(T)$ is given by (\ref{0x}).  					We have 
	$$		 \mathbb{E}^Q[X^{\pi^*,\lambda}(T){\cal Y}(T)]\geq \mathbb{E}^Q[L\mathbb{I}_{\{H(T)<H_{\varepsilon}^*\}}y_0(1+\phi)H(T)F(T)]=\hat x_\varepsilon y_0>x_0y_0,
		$$
			which is a contradiction. Similarly, 
			if $\frac{c_{z}}{y_0(1+\phi)}<H_{\varepsilon}^*\leq \frac{c_{\tilde{z}}}{y_0(1+\phi)}$ or $H_{\varepsilon}^*>\frac{c_{\tilde{z}}}{y_0(1+\phi)}$, then 
			$X^{\pi^*,\lambda}(T)$ is given by (\ref{x-3}) with $c_{\tilde{z}_0}=y_0(1+\phi)H_{\varepsilon}^*$
			or (\ref{x-4}) with $k_{\lambda}=y_0(1+\phi)H_{\varepsilon}^*$, respectively, we would again have
			$\mathbb{E}^Q[X^{\pi^*,\lambda}(T){\cal Y}(T)]>x_0y_0$, which is a contradiction. We conclude that there is no feasible solution if $x_0<\hat x_\varepsilon$. 
				\end{proof}

				\begin{remark}
										If $\varepsilon=0$, which requires $X^\pi(T)\geq L$ a.s., then $H^*_\varepsilon=+\infty$, $\hat x_\varepsilon=(1+\phi)\mathbb{E}^Q[LF(T)H(T)]$ and 
										$$X^{\pi^*,\lambda^*}(T)=(\theta+I_1(Y(T)))\mathbb{I}_{\{Y(T)<c_{\tilde{z}}\}}+L\mathbb{I}_{\{Y(T)\geq c_{\tilde{z}}\}}.$$ 
					If $\varepsilon=1$, which removes the quantile constraint, then $H^*_\varepsilon=0$, $\hat x_\varepsilon=0$ and 
					$$X^{\pi^*,\lambda^*}(T)=(\theta+I_1(Y(T)))\mathbb{I}_{\{Y(T)\leq c_z\}}.$$ 
				\end{remark}

\section{Algorithms} \label{sec_algo}
In this section we discuss three algorithms to solve  problem (\ref{util}) with conditions (\ref{lamd}). The first one is the exact algorithm based on Theorem \ref{thm_feasible}, the second the dual simulation method based on (\ref{dutils}), and the third the PINN method to solve the dual HJB equation (\ref{dhjbx}).

\subsection{Lagrange algorithm}
The proof of Theorem \ref{thm_feasible} is constructive, and from this we propose a numerical algorithm to solve the problem. We note that the values of $y_0$ are different in different $X^{\pi^*,\lambda^*}(T)$, which are denoted by $y_0^1$, $y_0^2$, $y_0^3$ in (\ref{0x}), (\ref{x-3}), and (\ref{x-4}) respectively. Now we give an algorithm to derive the optimal terminal wealth $X^{\pi^*, \lambda^*}(T)$ and the corresponding Lagrange multiplier  $\lambda^*$. The exact algorithm is the following.		

\begin{itemize}
\item[Step 0] Input initial wealth \( x_0 \), initial drift estimate $ \hat{\mu}(0) $, confidence level $ \varepsilon $, then compute $ z $ \eqref{z}, $ \tilde{z} $ \eqref{zt}, $ H^*_\varepsilon $ \eqref{eq_defH*}, and $\hat x_\varepsilon$ (\ref{x_epsilon}). If $x_0<\hat x_\varepsilon$, then the problem is infeasible, stop. If $x_0=\hat x_\varepsilon$, then set $X^{\pi^*,\lambda^*}(T)=L\mathbb{I}_{H(T)<H^*_\varepsilon}$, stop. Otherwise, go to Step 1. 
\item[Step 1] Compute $ y_0^3 $ from the budget constraint \eqref{bc} using $X^{\pi^*, \lambda}(T)$ from \eqref{x-4}. If $y_0^3 > \frac{c_{\tilde{z}}}{H^*_\varepsilon (1 + \phi)} $, then set 
 $X^{\pi^*, \lambda^*}(T) = X^{\pi^*, \lambda}(T)$ and $ \lambda^* = \lambda_2(y_0^3) $  from \eqref{l-2-A}, stop.  Otherwise, go to Step 2.
 \item[Step 2] Compute $ y_0^2 $ from \eqref{bc} using $X^{\pi^*, \lambda}(T)$ from \eqref{x-3}. 
If $y_0^2 > \frac{c_z}{H^*_\varepsilon (1 + \phi)} $, then 
set $X^{\pi^*, \lambda^*}(T) = X^{\pi^*, \lambda}(T)$ and $ \lambda^* = \lambda_1(y_0^2) $ from \eqref{l-1-A}, stop.   Otherwise, go to Step 3.
\item[Step 3] Set
$X^{\pi^*, \lambda^*}(T) $ from (\ref{0x})
and $ \lambda^* = 0 $, stop.
\end{itemize}



Since $\lambda^*$ can be expressed as a function of $y_0$ that is related to  $\varepsilon$, we write the optimal terminal wealth as $X^{\pi^*,\varepsilon}(T)$ instead of $X^{\pi^*,\lambda^*}(T)$  and characterize its form in terms of $\varepsilon$. 
Specifically,   denote by $\varepsilon^*=1-\mathbb{P}(H(T)\leq \frac{c_z}{y_0(1+\phi)})$ and $\varepsilon_*=1-\mathbb{P}(H(T)\leq \frac{c_{\tilde{z}}}{y_0(1+\phi)})$. Since $\tilde z<z$, we have $c_{\tilde z}>c_z$ and $\varepsilon_*\leq \varepsilon^*$. If $\varepsilon\geq \varepsilon^*$, then  $\mathbb{E}^Q[F(T)\mathbb{I}_{\{H(T)\leq \frac{c_z}{y_0(1+\phi)}\}}]\geq \mathbb{E}^Q[F(T)\mathbb{I}_{\{H(T)\leq H_\varepsilon^*\}}]=1-\varepsilon$, we have $H_\varepsilon^*\leq \frac{c_z}{y_0(1+\phi)}$ and $X^{\pi^*,\varepsilon}(T)$ is given by (\ref{0x}).
If $\varepsilon_*\leq \varepsilon< \varepsilon^*$, then $\mathbb{E}^Q[F(T)\mathbb{I}_{\{H(T)\leq \frac{c_{\tilde{z}}}{y_0(1+\phi)}\}}]\geq \mathbb{E}^Q[F(T)\mathbb{I}_{\{H(T)\leq H_\varepsilon^*\}}]>  \mathbb{E}^Q[F(T)\mathbb{I}_{\{H(T)\leq \frac{c_z}{y_0(1+\phi)}\}}]$, we have $H_\varepsilon^*\in (\frac{c_z}{y_0(1+\phi)},\frac{c_{\tilde{z}}}{y_0(1+\phi)}]$ and $X^{\pi^*,\varepsilon}(T)$ is given by (\ref{x-3}). 
If $\varepsilon<\varepsilon_*$, then $\mathbb{E}^Q[F(T)\mathbb{I}_{\{H(T)\leq H_\varepsilon^*\}}]>\mathbb{E}^Q[F(T)\mathbb{I}_{\{H(T)\leq \frac{c_{\tilde{z}}}{y_0(1+\phi)}\}}]$, we have $H_\varepsilon^*>\frac{c_{\tilde{z}}}{y_0(1+\phi)}$ and $X^{\pi^*,\varepsilon}(T)$ is given by (\ref{x-4}).

\subsection{Dual simulation algorithm}
First we consider running Monte Carlo simulations for the dual problem (\ref{dutils}). Due to the unconstrained, complete market nature of the problem, the dual problem is reduced to evaluation of an expectation, so Monte Carlo methods are well suited. The only difficulty is in converting to the primal problem, we need to find the dual start parameter $y^*$ associated to $x_0$, minimizing $v_\lambda^c(0, y, \hat{\mu}(0)) + x_0 y$, or equivalently solving $\partial_y v_\lambda^c(0, y, \hat{\mu}(0)) + x_0  = 0$. Combined with the need to find the Lagrange multiplier $\lambda^*$ leads to a coupled optimization problem to solve. Optimization over $y$ is very easy, as we can evaluate $\partial_y v_\lambda^c(0, y, \hat{\mu}(0))$ again using simulation. We therefore find $y^*$ for a range of $\lambda$ values, then can find the right $\lambda$ to solve the constrained problem. 

Fix a sample size $M \in \N$ and discretization size $N \in \N$. Let $h = \frac{T}{N}$ be the step size of the corresponding discretization on $[0, T]$.  Define the following Monte Carlo function
\begin{align}\label{eq_dual_mc}
\text{MC}^d(y, \lambda) = \frac{1}{M}\sum_{i = 1}^{M} V_\lambda^c\left(y \zeta^i_N\right),
\end{align}
where
\begin{align} \label{eq_dual_sim} \begin{split}
\zeta^i_{n+1} & = \zeta_n^i \left(1  - h r  - \sqrt{h}\sigma^{-1}(\hat{\mu}_n^i - r)  Z_n^i\right),\\
\hat{\mu}^i_{n+1} & = \hat{\mu}_n^i + \sqrt{h}\psi\left(nh, \hat{\mu}_n^i\right) Z_n^i,
\end{split}
\end{align}
for $i = 1, \ldots, M$ and $n = 0, \ldots, N - 1$, 
$\zeta^i_0  = 1$, $\hat{\mu}^i_0  = \hat{\mu}(0)$, 
and $Z^i_n$ are independent standard normal random variables. Given fixed $\lambda$, the optimal dual start $\hat{y}$ should (approximately, given sufficiently large $M$) minimize the function
$y \mapsto \text{MC}^d(y, \lambda) + x_0 y$ for $y > 0$. 
We therefore define the gradient descent update, for $k \in \N$
\begin{align*}
y_{k+1} & = y_k - \delta (\partial_y \text{MC}^d(y_k, \lambda) + x_0) \\
& = y_k - \delta \left(\frac{1}{M}\sum_{i = 1}^{M} \zeta^i_N(V^c_\lambda)'\left(y_k \zeta^i_N\right) + x_0\right)
\end{align*}
for some learning rate $\delta > 0$ and initial $y_0 > 0$. Assuming convergence of the algorithm to some $y^* = \lim_{k \to \infty}y_k$, we output the primal Monte Carlo simulation 
\begin{align*}
u_{\lambda}^c(t, x_0, \hat{\mu}) & = \frac{1}{M}\sum_{i = 1}^{M} U_\lambda^c\left(x^{*, \lambda}\left(y^* \zeta^i_N\right)\right), \\
h_{\lambda}(t, x_0, \hat{\mu}) & = \frac{1}{M}\sum_{i = 1}^{M} \ind_{x^{*, \lambda}\left(y^* \zeta^i_N\right) \geq L}.
\end{align*}
This method applies for only one value of $\lambda$. However, the algorithm can be run in parallel for multiple values $(\lambda_j)_{j=1}^J \subset [0, \infty)$. We can then interpolate between these points to produce $u$ and $h$ as functions of $\lambda$.

\begin{remark} \label{rem_infeasible} The training range may intersect with the infeasible region, points $x, \hat{\mu}$ at which the constrained problem (\ref{uo}) starting at $x, \hat{\mu}$ has no solution for some $\varepsilon \in [0, 1]$. However, the unconstrained problem (\ref{util-1}) is well defined and has a solution for any start points $x \geq 0, \hat{\mu} \in [\mu^l, \mu^h]$ and multiplier $\lambda \geq 0$. This is also the case in the next algorithm.
\end{remark}

\subsection{Dual PINN algorithm}

By Feynmann Kac, the dual value function (\ref{eq_dual_value}) satisfies a PDE (\ref{dhjbx}) which can be solved numerically. In particular, due to the simple nature of the dual problem, the dual PDE is linear, suggesting high effectiveness of numerical methods. The PINN method assumes a bounded state space, and samples points in this space. The dual value functions is defined to be a neural network, to which the PDE operator can be applied. This term is combined with the error at terminal time, to produce a loss function that can be minimized over the parameters of the neural network.

Typical convergence analysis for the PINN method requires PDEs with uniform Lipschitz continuity of the terminal function \cite{shin2020convergence}. For the HJB equation satisfied by the dual value function, the terminal function is only a Lipschitz continuous function of $y$ away from $y = 0$. Therefore, even in the simple case of the dual problem, our PDEs fall out of the provably convergent class of PDEs. However, we can still attempt this method and compare to the simulation approach. Just like the dual simulation method, we need to find the optimal dual start after finding the dual value function, in conjunction with the Lagrange multiplier.

First we define a neural network. Let ${\cal N} \in \N$ be fixed, this is the number of ``nodes'' in our network. Let $f \colon \R^d \times \Xi^{d, {\cal N}} \to \R$ for some input dimension $d \in \N$, where $\Xi^{d, {\cal N}} = \R^{d \times {\cal N}} \times \R^{{\cal N} \times {\cal N}} \times \R^{{\cal N} \times 1} \times \R^{\cal N} \times \R^{\cal N} \times \R$ is the parameter space. For $X \in \R^d$ and $\Theta = (A_1, A_2, A_3, b_1, b_2, b_3) \in \Xi^{d, {\cal N}}$ we define the neural network $f$ as 
\[f(X; \Theta) = L_3 \circ \eta \circ L_2 \circ \eta \circ L_1  (X), \qquad X \in \R^d\]
where $L_i(X) = A_i X + b_i$ are linear maps ($A$ is the ``weight'', $b$ is the ``bias'') for $i =1, 2, 3$ and $\eta$ is a non-linear ``activation'' function, applying the function $x \mapsto \tanh(x)$ element-wise. Each output of $\eta \circ L_i$, and the input $X$ are referred to as a layer. We have the input layer, the output layer $f(X; \Theta)$, and the two so-called ``hidden layers'' in between. In the sequel we will take $d = 4$.

The dual PINN method aims to find the value of $v^c_\lambda(t, y, \hat{\mu})$ and  for $t \in [0,T]$, $y \in \Y \subset [0,\infty)$, $\hat{\mu} \in \M \subset [\mu^l,\mu^h]$, and $\lambda \in \Lambda \subset [0,\infty)$. The function $v(t, y, \hat{\mu}, \lambda) \defeq v^c_\lambda(t, y, \hat{\mu})$ solves the PDE (\ref{dhjbx}). We can solve the HJB equation using the PINN method to find the function $v^c$. We can then solve the optimality condition for $y^* \defeq y^*(\lambda)$ at time 0 for a range of $\lambda$. We can then evaluate 
\[g_\lambda(0, y^*(\lambda), \hat{\mu}(0)) = \frac{1}{M}\sum_{i = 1}^M  \ind_{x^{*,\lambda}(y^*(\lambda)\zeta_N^i) \geq L}
 \]
 using Monte Carlo simulation (\ref{eq_dual_sim}). 
 

We initialize a neural network function $v(t, y, \hat{\mu} , \lambda; \Theta^v)$ depending on some parameters $\Theta^v$. We define the following loss functionals for any twice differentiable $v  \colon [0,T] \times \Y \times \M \times \Lambda \to \R$.
\begin{align}\begin{split} \label{eq_dual_pinn_loss}
\LLL(v) & = \sum_{i = 1}^{K_c} \bigg|\partial_t v(t^c_i, y^c_i, \mu^c_i, \lambda^c_i) - r y^c_i \partial_y v(t^c_i, y^c_i, \mu^c_i, \lambda^c_i)  + \half (y^c_i)^2\left|\sigma^{-1}(\mu^c_i - r)\right|^2 \partial_{yy} v(t^c_i, y^c_i, \mu^c_i, \lambda^c_i) \\
& \qquad \, + \half \psi(t^c_i, \mu^c_i)^2 \partial_{\mu \mu}v(t^c_i, y^c_i, \mu^c_i, \lambda^c_i)- \sigma^{-1}y^c_i(\mu^c_i - r\ind) \psi(t^c_i, \mu^c_i) \partial_{y \mu} v(t^c_i, y^c_i, \mu^c_i, \lambda^c_i)\bigg|^2, \\
\LLL_V(v) &= \sum_{i = 1}^{K_b} \big|v(T, y^b_i, \mu^b_i, \lambda^b_i) - V^c_{\lambda^b_i}(y^b_i)\big|^2, \\ 
\end{split}
\end{align}
where $t^c_i, y^c_i, \mu^c_i, \lambda^c_i$ for $i = 1, \ldots, K_c$ are generated uniformly from $[0,T] \times \Y \times \M \times \Lambda$ (``collocation points'') and $y^b_i, \mu^b_i, \lambda^b_i$ for $i = 1, \ldots, K_b$ are generated uniformly from $\Y \times \M \times \Lambda$ (``boundary points'').
We minimize the sum of these functions using gradient descent, starting at some arbitrary $\Theta^v_0$, evaluated element wise along $\Theta$
\begin{align*}
\Theta^v_{k+1} & = \Theta^v_k - \delta \partial_\Theta \left[\LLL(v(\cdot; \Theta^v_k)) + \LLL_V(v(\cdot; \Theta^v_k))\right],
\end{align*}
for $k \in \N$ and some learning rate $\delta > 0$. Assuming convergence, let $\Theta^v$ be the converged parameter set.
For the dual problem, we need to find the optimal dual start $y^*$ as well as the optimal Lagrange multiplier $\lambda^*$ for fixed $x_0$ and $\hat{\mu}$, which solve the coupled optimization problem
\begin{align} \begin{split} \label{eq_pinn_opt}
y^* & = \arg \min_{y > 0}\left\{ v_{\lambda^*}^c(0, y, \hat{\mu}) + x_0 y \right\}, \\
\lambda^* & = \arg \min_{\lambda \geq 0} \left\{\lambda \left|g_{\lambda}^c(0, y^*, \hat{\mu}) - (1 - \varepsilon)\right| + \max\left(0,(1 - \varepsilon) - g_{\lambda}^c(0, y^*, \hat{\mu}) \right)  \right\}.
\end{split}
\end{align}
To solve this problem, we first find $y^*$ solving 
$\partial_y v(t, y^*, \hat{\mu}, 0 ; \Theta^v) + x_0 = 0.$
We can then test if $g_0(t, y^*, \hat{\mu}) \geq (1 - \varepsilon)$, in which case $\lambda^* = 0$ and we are done. Otherwise, we are searching for $y^*, \lambda^* > 0$ such that $\partial_y v(t, y^*, \hat{\mu}, \lambda^* ; \Theta^v) + x_0 = 0$ and $g_{\lambda^*}(0, y^*, \hat{\mu}) - (1 - \varepsilon) = 0 $. We can use existing numerical optimization algorithms that take the functions $v(t, y, \hat{\mu}, \lambda ; \Theta^v)$ and $g_\lambda(t, y, \hat{\mu})$ and solve
\[|\partial_y v(0, y^*, \hat{\mu}, \lambda^* ; \Theta^v) + x_0|^2 + |g_{\lambda^*}(0, y^*, \hat{\mu}) - (1 - \varepsilon)|^2 = 0,\]
over the trained ranged $y, \lambda \in \Y \times \Lambda$. If we have either $y^* \in \partial \Y$ or $\lambda^* \in \partial \Lambda \backslash \{0\}$ on the boundary on the training region, then it is likely we have not found the true optimizers and we should increase the training range and repeat training. 

If we are only interested in finding $y^*$ and the corresponding value and constraint for fixed $t, x, \hat{\mu}, \lambda$, we only need to solve  
$\partial_y v(t, y^*, \hat{\mu}, \lambda ; \Theta^v) + x = 0.$
This can again be solved numerically, and we then output
\begin{align*}
u_{\lambda}^c(t, x, \hat{\mu}) & = v(t, y^*, \hat{\mu}, \lambda ; \Theta^v) + x y, \\
h_{\lambda}(t, x, \hat{\mu}) & = g_\lambda(t, y^*, \hat{\mu} ).
\end{align*}
In practise, we find a $y^*(\lambda)$ for each $\lambda$, and then find $\lambda^*$ such that ($\lambda^*, y^*(\lambda^*)$) satisfies (\ref{eq_pinn_opt}). We do this because we have an explicit representation for the neural network $v$ and it's derivative, so can easily optimize using it, unlike $g_\lambda$ which is evaluated via simulation. Also, once we have $y^*(\lambda)$, we simulate the optimal terminal wealth via $X^{\pi^*, \lambda}_N = x^{*, \lambda}(y^*(\lambda)\zeta_N)$ using (\ref{eq_dual_sim}).

\section{Numerical Examples} \label{sec_numerics}
In this section we solve the constrained problem (\ref{util}) with all methods. For the algorithms that solve the unconstrained problem,
given fixed $\varepsilon > 0$, we numerically solve the problem (\ref{util}) in two steps. Firstly, we set $\lambda = 0$ and find the values of $u_{0}^c(t, x, \hat{\mu})$ and $h_{0}(t, x, \hat{\mu})$ given by (\ref{eq_value}) and (\ref{eq_constraint}) respectively. If $h_{0}(t, x, \hat{\mu}) \geq 1 - \varepsilon$ then we output $u_{0}^c(t, x, \hat{\mu})$ and $\lambda^* = 0$ and we are done. Otherwise, we find $u_{\lambda}^c(t, x, \hat{\mu})$ and $h_{\lambda}(t, x, \hat{\mu})$ for a sufficiently large range of $\lambda > 0$. We then find $\lambda^*$ such that $h_{\lambda^*}(t, x, \hat{\mu}) = 1 - \varepsilon$.

\subsection{Data} \label{sec_config}
Unless otherwise mentioned, we take $T = 1.0$, $d = 1$, $\theta = 1.5$, $L = 0.9$, $x_0 = 1.0$, $r = 0.05$, $\sigma = 0.2$, $\lambda = 0.2$. We use the utility $U_1(x) = \sqrt{x}$, $U_2(x) = x^{0.3}$. With this configuration, problem (\ref{uo}) admits a unique solution for any $\varepsilon \in [0,1]$ by Theorem \ref{thm_feasible}. In this setting we have
\begin{align*}
\tilde{z}_0 & = \theta + \left(\sqrt{\left(U_2(\theta) + \lambda\right)^2 + \theta} - \left(U_2(\theta) + \lambda\right)\right)^2, \\
\tilde{z} & = \theta + \left(\sqrt{U_2(\theta - L)^2 + \left(\theta - L\right)} - U_2(\theta - L)\right)^2.
\end{align*}

In all algorithms, we use grids $\X = [0.2, 2.0]$, $\Y = [0.2, 2.0]$, $\Lambda = [0.0, 2.5]$, and $\M = [0.03, 0.1]$. We use a grid of $(\lambda_j)_{j = 1}^J \subset \Lambda$ equally spaced points in $\Lambda$ with $J = 51$ (we choose 50 subintervals over the region, then $J$ accounts for the midpoints, including both end points). For neural networks, we use a neural network structure with 2 hidden layers, with tanh activation function. For the primal simulation method the network layers have 10 hidden nodes, and for the dual PINN method they have 100 nodes. 
For dual simulation we take $N = 100$, $M  = 100000$, and run the algorithm for 200 steps with $\delta = 0.1$.
We run the PINN algorithm generating 2000 and 200 collocation and boundary points respectively, and use $\delta = 0.01$. We run until either 100000 iteration steps, or the loss function is below $0.00005$. The subsequent constrained optimization is performed using \texttt{Scipy}, using the neural network function as input into the equation. 
We solve the discrete distribution problem, taking $\psi(u)=\sigma^{-1}(u-\mu^l)(\mu^h-u){}\ind_{\mu^l\leq u\leq \mu^h}$ to facilitate comparison with the exact algorithm.
The neural networks are implemented using \texttt{Tensorflow} and parameters are optimized using the \texttt{ADAM} algorithm. Code to implement the algorithms can be found in \texttt{https://github.com/Ashley-Davey/ML-For-Quantile}.


\subsection{Numerics}

Table \ref{tab_results} displays some statistics for specific values of $\varepsilon$. By inverting the primal or dual constraint functions, we compute the value and other information for a single value of $\varepsilon$. We give the values of $\lambda^*$ and $y^*$, along with the concavified problem value $u^c_{\lambda^*}(0, x_0, \hat{\mu}(0))$ \eqref{eq_value}, the ``true value'' $u(0, x_0, \hat{\mu}(0)) \defeq \mathbb{E}[U(X^{\pi^*, \lambda^*}(T))|X^{\pi^*, \lambda^*}(0)=x,{\hat \mu}(0)={\hat \mu}]$ as a Monte Carlo evaluation of the utility at the outputted optimal state $X^{\pi^*, \lambda^*}(T)$ of the Lagrange problem with optimal Lagrange multiplier, and the probability that this state matches the lower limit and 0 exactly.

\begin{table}
\centering																
\begin{tabular}{|c|c|c|c|c|c|c|c|}\hline																
$\varepsilon$ 	&	Method	&	 $\lambda^*$	&	 $y^*$	&	    $u$	&	 $u^c_{\lambda^*}$	&	$\P(X^{\pi^*, \lambda^*}(T)=L)$	&	$\P(X^{\pi^*, \lambda^*}(T)=0)$	\\ \hline	\hline  
0	&	Dual sim	&	1.65	&	1.883	&	-0.561	&	1.086	&	0.746	&	0.003	\\ \hline	
0	&	Dual pinn	&	1.7	&	1.911	&	-0.601	&	1.13	&	0.782	&	0	\\ \hline		
0	&	Lagrange	&	1.659	&	1.885	&	-0.564	&	1.095	&	0.751	&	0	\\ \hline	\hline
0.1	&	Dual sim	&	1.453	&	1.795	&	-0.411	&	0.898	&	0.5	&	0.102	\\ \hline	
0.1	&	Dual pinn	&	1.464	&	1.806	&	-0.435	&	0.91	&	0.52	&	0.1	\\ \hline		
0.1	&	Lagrange	&	1.452	&	1.794	&	-0.411	&	0.896	&	0.5	&	0.1	\\ \hline	\hline
0.35	&	Dual sim	&	0.478	&	1.214	&	-0.095	&	0.216	&	0	&	0.35	\\ \hline	
0.35	&	Dual pinn	&	0.614	&	1.295	&	-0.114	&	0.31	&	0	&	0.35	\\ \hline		
0.35	&	Lagrange	&	0.483	&	1.216	&	-0.095	&	0.219	&	0	&	0.35	\\ \hline	\hline
1	&	Dual sim	&	0	&	0.946	&	-0.086	&	-0.085	&	0	&	0.395	\\ \hline	
1	&	Dual pinn	&	0	&	0.94	&	-0.046	&	-0.084	&	0	&	0.374	\\ \hline		
1	&	Lagrange	&	0	&	0.945	&	-0.085	&	-0.085	&	0	&	0.395	\\ \hline	
\end{tabular}						
\caption{Various statistics for different values of $\varepsilon$, applied at $t = 0$, $x_0 = 1.0$ and $\hat{\mu}(0) = 0.07$.}
\label{tab_results}
\end{table}

Figure \ref{fig_full} shows the results for the constrained problem for all methods. For each method, we find the value and constraint function for a range of $\lambda$, and take the (right) inverse the graph of $\lambda \mapsto h_\lambda(0, x_0, \hat{\mu}(0))$ to get a mapping $\varepsilon \mapsto \lambda^*(\varepsilon)$ for $\varepsilon \leq \varepsilon_0 \defeq 1 - \P(X^{\pi^*, 0}(T) \geq L)$. For $\varepsilon > \varepsilon_0$, the solution of the constrained problem is equal to the solution at $\varepsilon = \varepsilon_0$ as the quantile constraint is non-binding. Figure \ref{fig_full} (a) shows the graph of $\varepsilon$ against Lagrange multiplier $\lambda^*$,  (b)  the corresponding primal value $\E[U^c_\lambda(X^{\pi^*, \lambda}(T))]$, and (c) the constraint probability $\P(X^{\pi^*, \lambda}(T) \geq L)$.   We run the algorithms 10 times and take an average. Figure \ref{fig_full} (b) is made up of two sections, the kink point is the point at which the concavified utility moves from one line segment to two, and the structure changes much faster in terms of $\varepsilon$. The simulation and PINN values appear to agree, with a slight gap between the two problems accounting for numerical error.
Figure \ref{fig_full} (c)  verifies that in the binding region $\varepsilon \in [0, \varepsilon_0]$ we have $\varepsilon = 1 - \P(X^{\pi^*, \lambda^*}(T) \geq L)$, and after this point we have $\P(X^{\pi^*, \lambda^*}(T) \geq L) = 1 - \varepsilon_0$ and $\lambda^* = 0$. The differences in the graph correspond to the differences in each algorithms' approximation of $\varepsilon_0$, which is where each graph of Figure \ref{fig_full} (a) hits the x-axis, or the value of $\P(X^{\pi^*, \lambda^*}(T)=0)$ when $\varepsilon = 1$ in Table \ref{tab_results}.

\begin{figure}[H] 
\centering
\begin{minipage}{.31\textwidth}
\centering
\begin{subfigure}[b]{\textwidth}
\includegraphics[width=\textwidth]{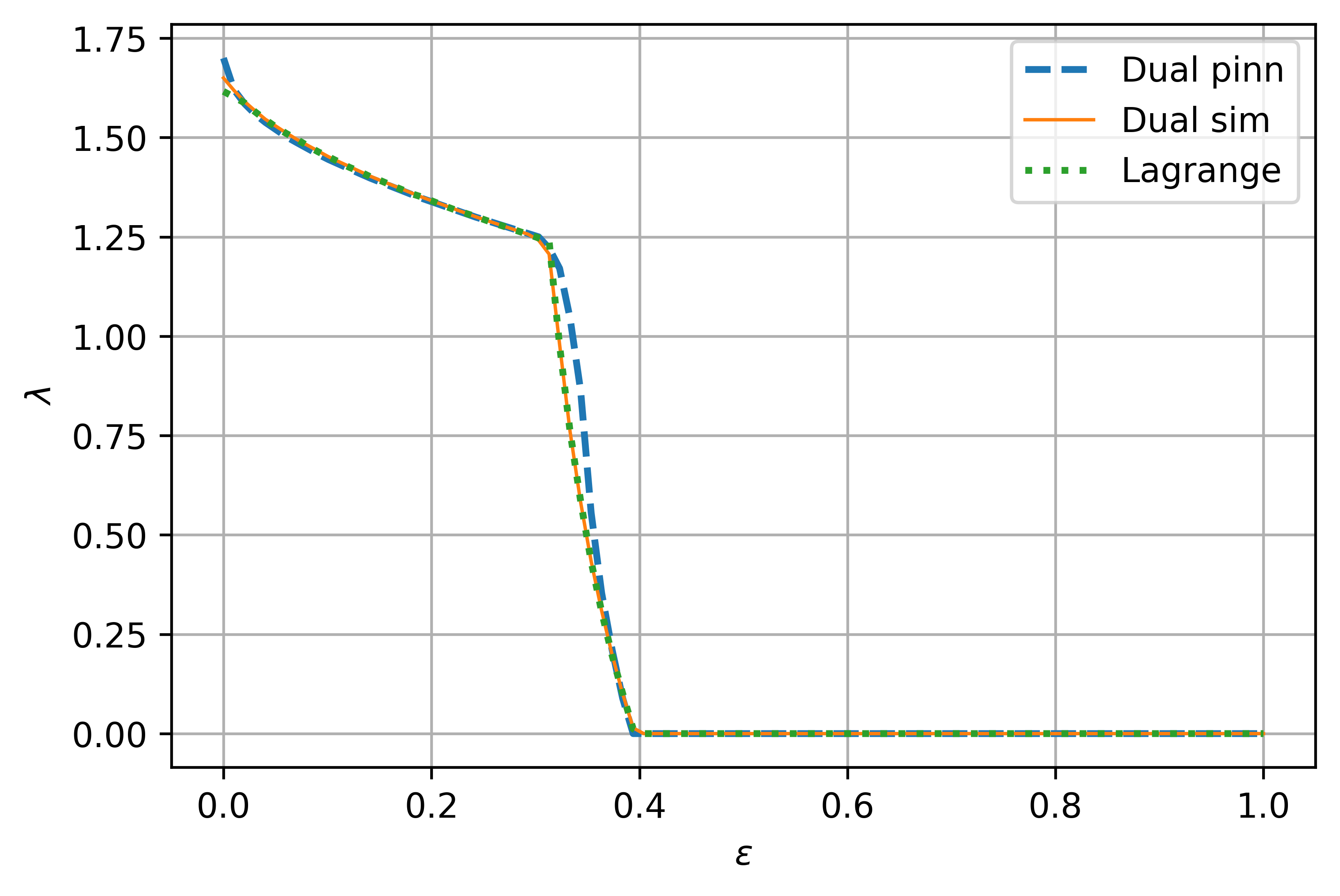}
\caption{$\varepsilon$ against $\lambda^*$. }
\end{subfigure}
\end{minipage}%
\;\;
\begin{minipage}{.31\textwidth}
\centering
\centering
\begin{subfigure}[b]{\textwidth}
\includegraphics[width=\textwidth]{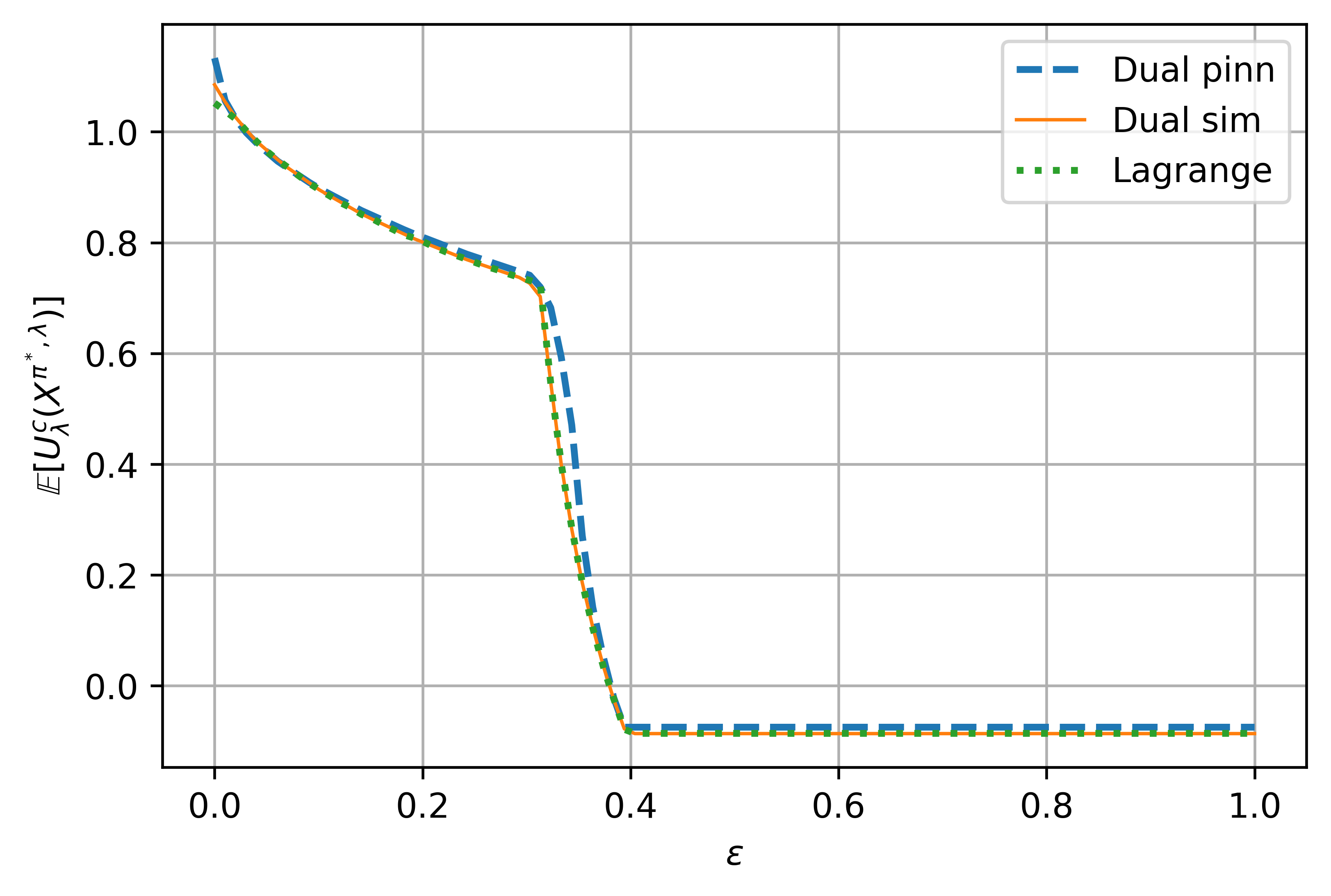}
\caption{$\varepsilon$ against $\E[U_\lambda^c(X^{\pi^*, \lambda^*}(T))]$.}
\end{subfigure}
\end{minipage}
\;\;
\begin{minipage}{.31\textwidth}
\centering
\centering
\begin{subfigure}[b]{\textwidth}
\includegraphics[width=\textwidth]{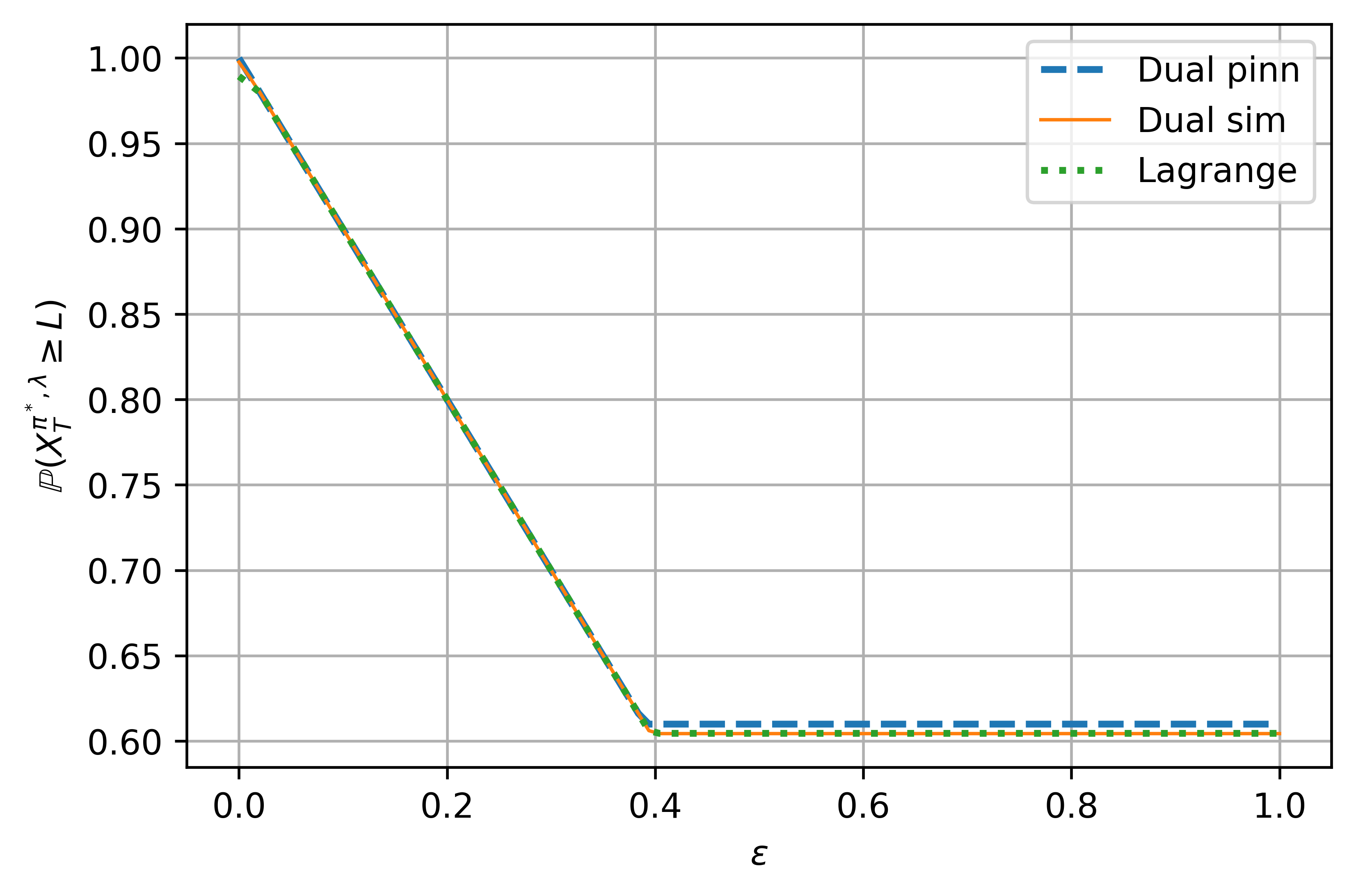}
\caption{$\varepsilon$ against $\P(X^{\pi^*, \lambda^*}(T) \geq L)$}
\end{subfigure}
\end{minipage}
\caption{Numerical results for the concavified problem.}
\label{fig_full}
\end{figure}


Figure \ref{fig_dist} shows the distribution of the optimal terminal wealth $X^{\pi^*, \lambda}(T)$ at terminal time when $\lambda \in \{0, 1.5, 2.5\}$. For the dual PINN and dual simulation methods, we generate this graph by simulating the dual state process and applying the function $x^{*, \lambda}$ given in Proposition \ref{dong}. Where there is an atom at $x = 0$ and $x = L$, we separate the distributions of each algorithm to make them clearer, but they all still refer to the same atom. We see the concavification principle applies, with the terminal state taking values at points
$x$ where $U_\lambda(x) = U^c_\lambda(x)$. The continuous section of the distributions have an exponentially decreasing tail in all cases, indicated by the linear segment of the log-scaled graphs.

\begin{figure}[H] 
\centering
\begin{minipage}{.31\textwidth}
\centering
\begin{subfigure}[b]{\textwidth}
\includegraphics[width=\textwidth]{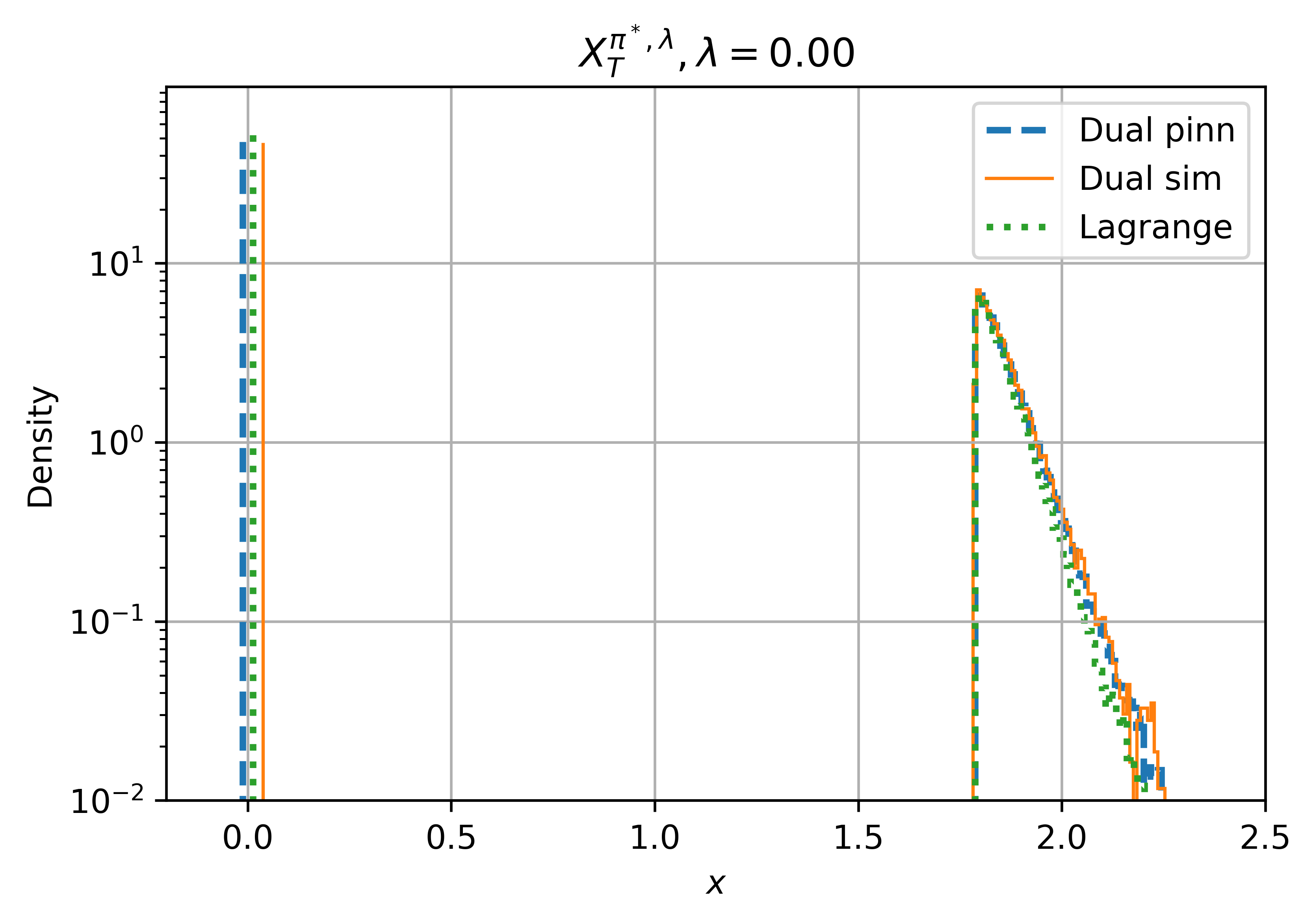}
\caption{$\lambda = 0$.}
\end{subfigure}
\end{minipage}%
\begin{minipage}{.31\textwidth}
\centering
\centering
\begin{subfigure}[b]{\textwidth}
\includegraphics[width=\textwidth]{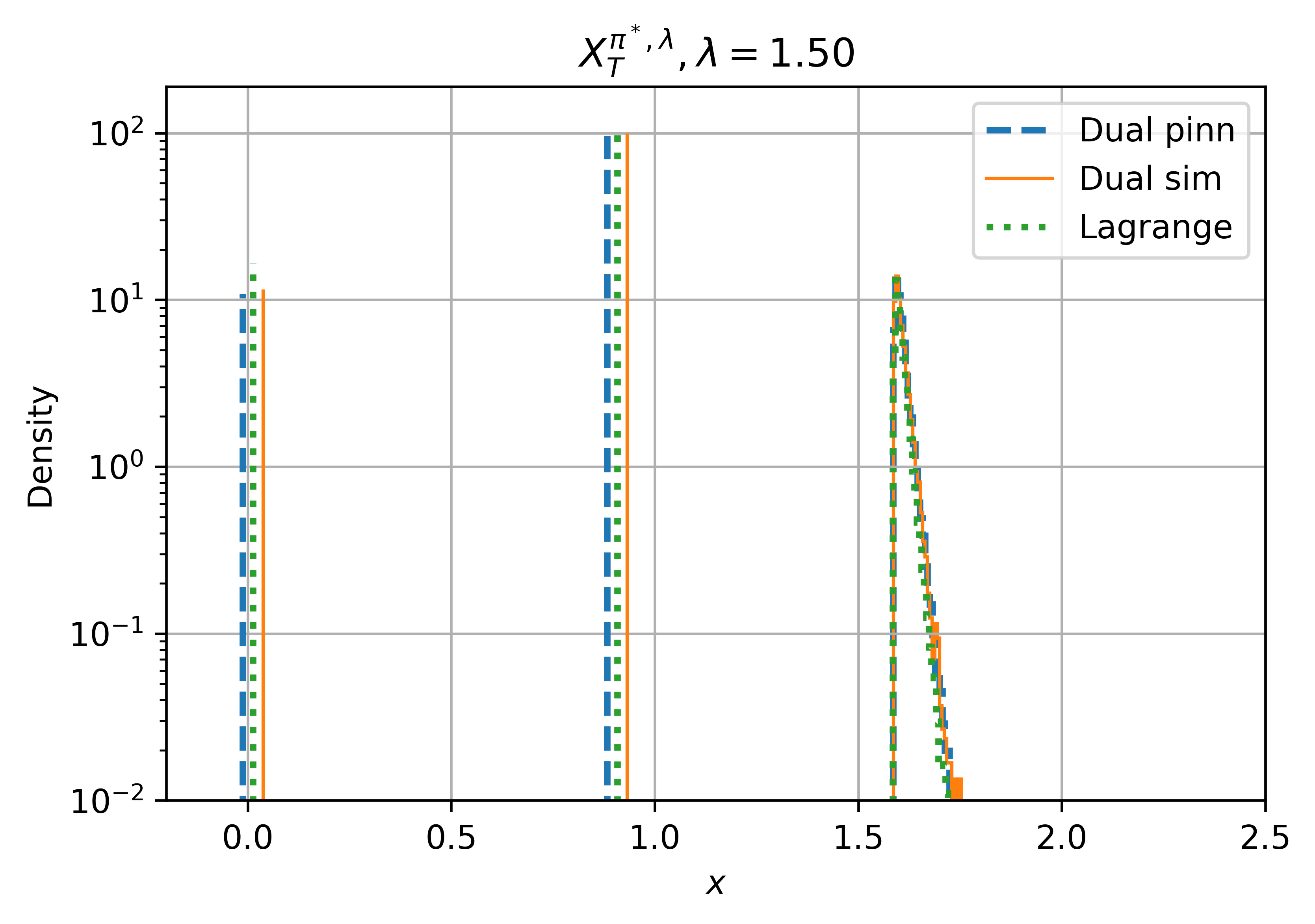}
\caption{$\lambda = 1.5$.}
\end{subfigure}
\end{minipage}
\begin{minipage}{.31\textwidth}
\centering
\centering
\begin{subfigure}[b]{\textwidth}
\includegraphics[width=\textwidth]{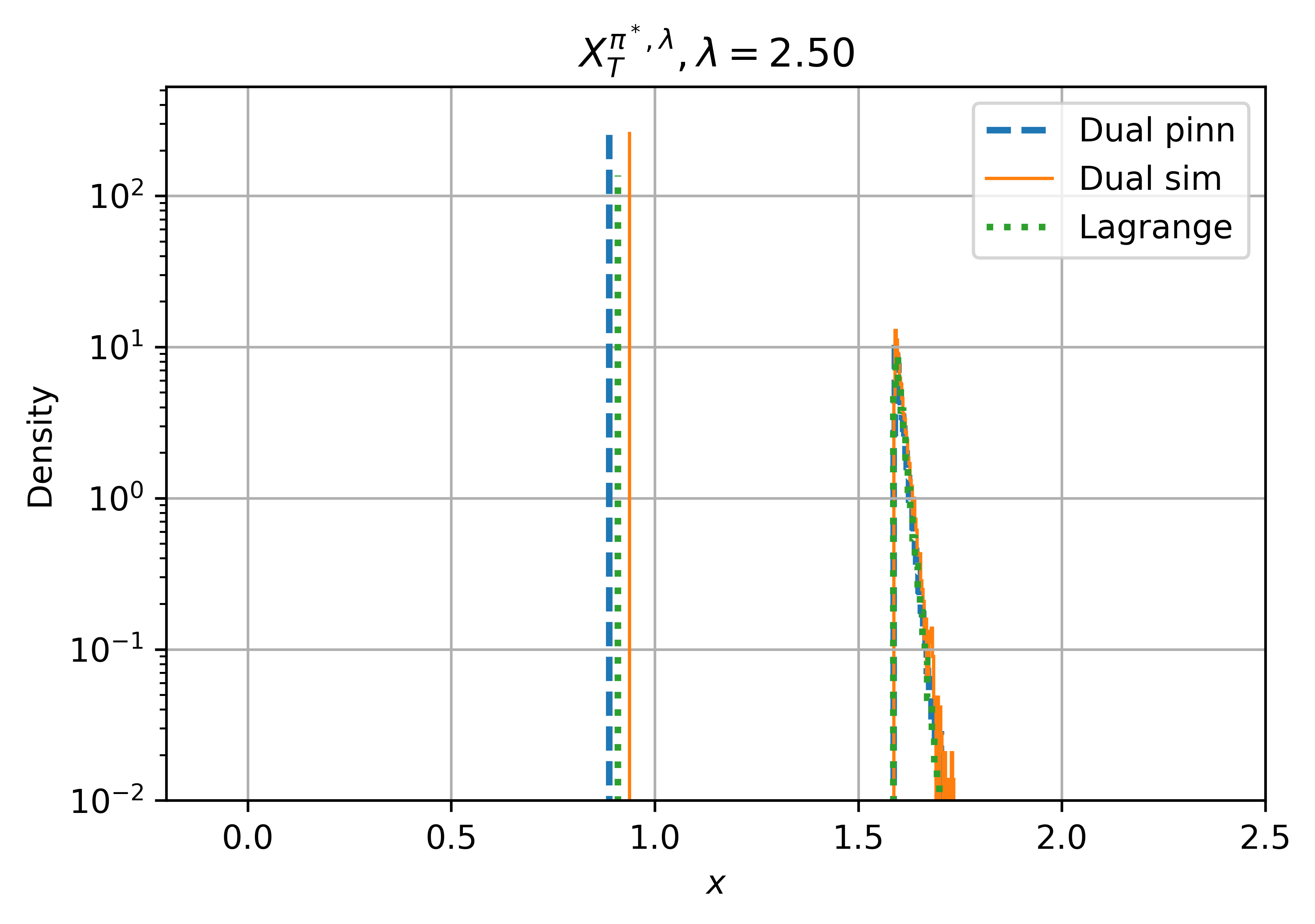}
\caption{$\lambda = 2.5$}
\end{subfigure}
\end{minipage}
\caption{ Distribution of $X^{\pi^*, \lambda}(T)$ for three different values of $\lambda$, log scaled.} 
\label{fig_dist}
\end{figure}

\subsection{Problem feasibility}

For fixed $\varepsilon = 0.2$, we may compute the value of $\hat{x}_\varepsilon$ defined in Theorem \ref{thm_feasible} as $\hat{x}_\varepsilon \approx 0.66$. This is the minimum wealth needed to achieve $\P(X^{\pi^*, \lambda}(T) \geq L) = 0.8$ for some $\lambda \geq 0$. Figure \ref{fig_algo2} (a) plots  the constraint function against $\lambda$ for feasible $x = 0.73$, infeasible $x = 0.6$ and the threshold $x = \hat{x}_\varepsilon \approx 0.66$. We see that $\hat{x}_\varepsilon$ is exactly the point such that $\lim_{\lambda \to \infty} \P(X^{\pi^*, \lambda}(T) \geq L) = 1 - \varepsilon$, where $X^{\pi^*, \lambda}(T)$ maximizes the value function $u^c_\lambda$. For $x < \hat{x}$ the constraint function is lower than $ 1 - \varepsilon$ for all $\lambda$, so there is no solution to the constrained problem.
With the Lagrange multiplier method, the objective function of the  optimization problem is given by, for fixed $\varepsilon \in [0,1]$ 
\begin{align*}
 J(\pi, \lambda; t, x, \hat{\mu}) & \defeq \E\left[U(X^\pi(T)) + \lambda \left(\ind_{X^\pi(T)\geq L} - (1 - \varepsilon)\right)\middle|X^{\pi}(t)=x,{\hat \mu}(t)={\hat \mu}\right] \\
 & =  \E\left[U_\lambda^c(X^\pi(T))\middle|X^{\pi}(t)=x,{\hat \mu}(t)={\hat \mu}\right] - \lambda (1 - \varepsilon).
 \end{align*}
Prior to this, the final term $- \lambda (1 - \varepsilon)$ has been omitted as it does not play a role in the optimization of $J$ over $\pi$. Define the corresponding (unconstrained) full value function
\[\bar{u}^c_\lambda( t, x, \hat{\mu}) =  u^c_\lambda(t, x, \hat{\mu}) - \lambda (1 - \varepsilon),\]
where the optimal $\pi^*$ depending on $\lambda$ has been found. We use the dual simulation algorithm to plot this function at an infeasible $x = 0.6$, feasible $x = 0.8$, and the transition point $x = \hat{x} \approx 0.66$ at $t = 0$. This plot is given in Figure \ref{fig_algo2} (b). If a solution exists, then it is a saddle point of the function $\lambda \mapsto \bar{u}^c_\lambda( 0, x, \hat{\mu})$. In the feasible region there is a unique saddle point at the solution, in the infeasible region there is no solution, and in the midpoint there is a limiting saddle point at $\lambda^* = \infty$ which corresponds to the only feasible wealth $X^{\pi^*, \infty}(T) \defeq L \ind_{H_T < H_\varepsilon^*}$.

\begin{figure}[H] 
\centering
\begin{minipage}{.31\textwidth}
\centering
\begin{subfigure}[b]{\textwidth}
\includegraphics[width=\textwidth]{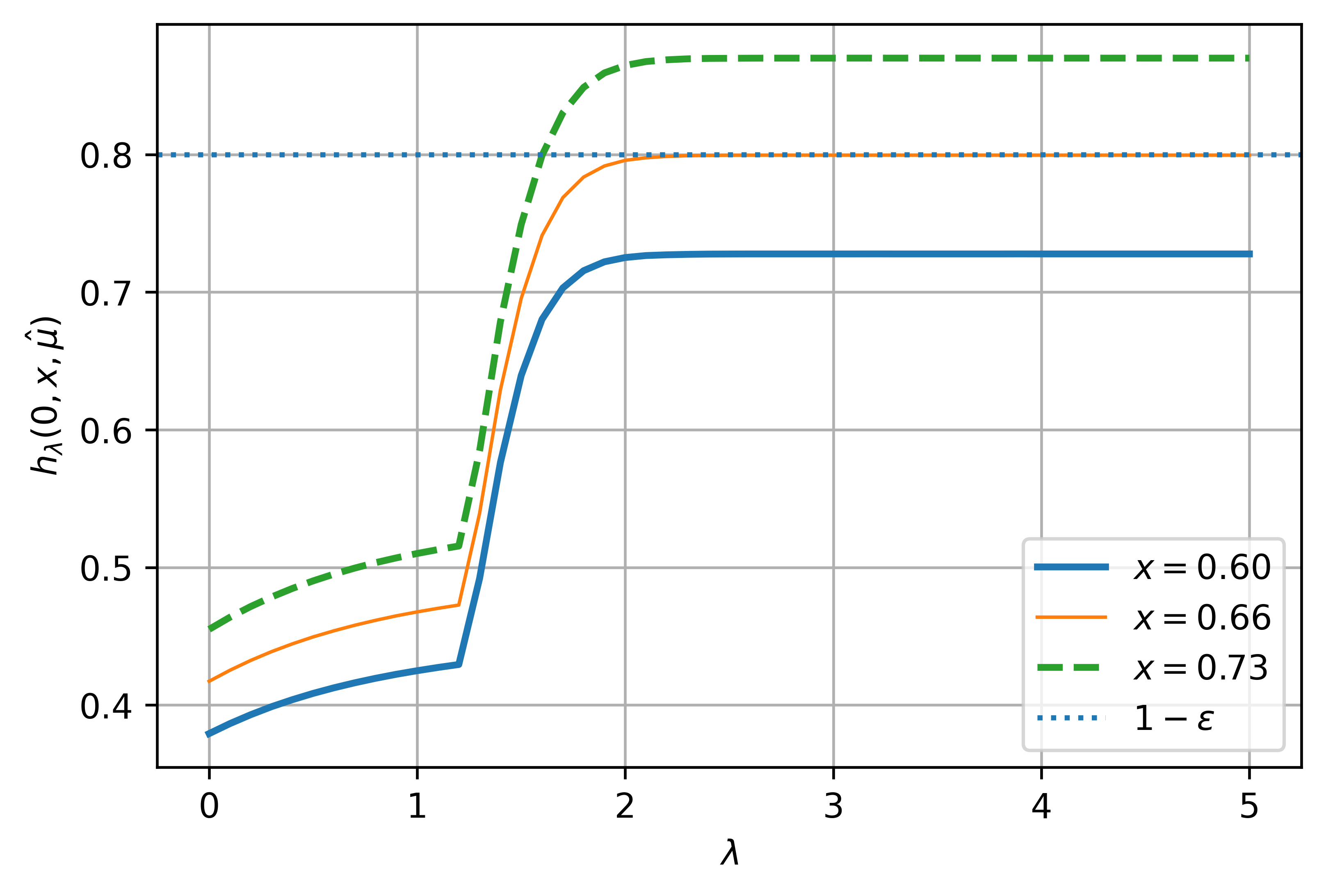}
\caption{$\lambda$ against constraint probability for 3 values of $x$.}
\end{subfigure}
\end{minipage}%
\;\;
\begin{minipage}{.31\textwidth}
\centering
\centering
\begin{subfigure}[b]{\textwidth}
\includegraphics[width=\textwidth]{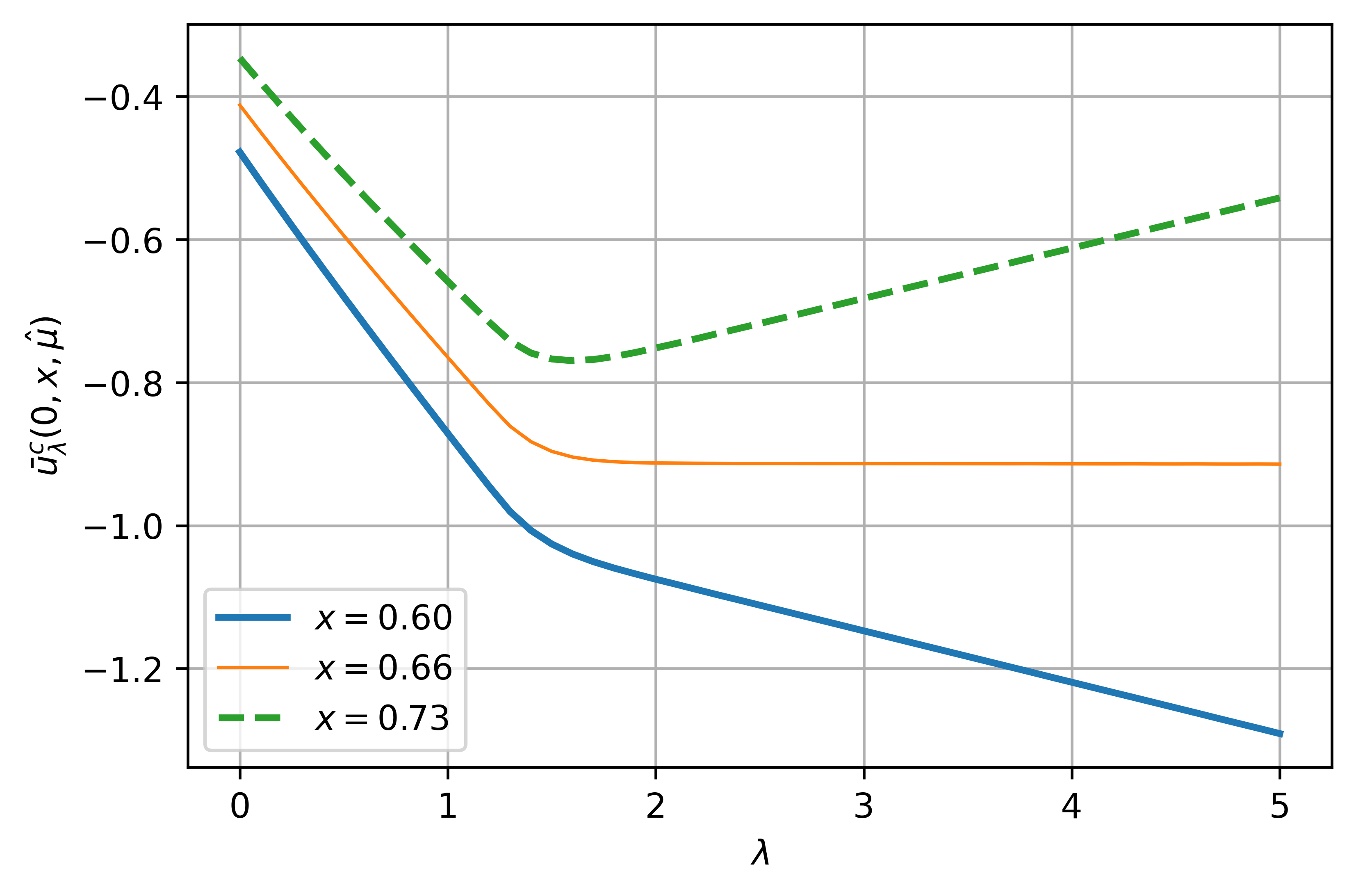}
\caption{$\lambda$ against full value function for 3 values of $x$.}
\end{subfigure}
\end{minipage}
\caption{Results for dual simulation and discrete algorithm in both feasible and infeasible regions.}
\label{fig_algo2}
\end{figure}

\section{Conclusions} \label{sec_conc}
\setcounter{equation}{0}

In this paper we solve S-shaped utility maximization incorporating both partial information and VaR constraint. We  convert the original unobservable model into an equivalent fully observable one with an additional filtered state process.  We then solve the problem in two steps,  first, solve an  unconstrained problem with the concavification principle and the dual method, and second, find the Lagrange multiplier and the initial dual state value for the constrained problem. We use a change of measure approach to overcoming the difficulty of the unknown joint distribution of the dual and filtered state processes and characterize the dual value function in a semi-closed integral form.  We identify a critical wealth level that makes the constrained problem admits a unique optimal solution or is infeasible. We also propose three algorithms (Lagrange, simulation, deep neural network) to numerically solve the problem and compare their performances with numerical examples. There remain many open questions. For example, if unobservable drift follows a general prior distribution, not necessarily a Bernoulli distribution, the current change of measure approach no longer works, how can we solve such a problem? We leave this and other questions for future research.

\bigskip\noindent{\bf Data Availability Statement.}
Data sharing not applicable - no new data generated, as the article describes entirely theoretical research.


\begin{thebibliography}{99}
	
	
\bibitem[\protect\citeauthoryear{Bain and Crisan}{2009}]{Bain2009} Bain, A. and Crisan, D. (2009). Fundamentals of Stochastic Filtering, Springer, New York.

\bibitem[\protect\citeauthoryear{Basak and Shapiro}{2001}]{basak01}
Basak, S., Shapiro, A. (2001). Value-at-risk-based risk management: optimal policies and asset prices. The review of financial studies, 14(2), 371-405.




\bibitem[\protect\citeauthoryear{Bensoussan et al.}{2022}]{ben22}
Bensoussan, A., Hoe, S., Kim, J., Yan, Z. (2022). A risk extended version of merton’s optimal consumption and portfolio selection. Operations Research, 70(2), 815-829.

\bibitem[\protect\citeauthoryear{Berkelaar et al.}{2004}]{bkp}
Berkelaar, A. B., Kouwenberg, R., Post, T. (2004). Optimal portfolio choice under loss aversion. Review of Economics and Statistics, 86(4), 973-987.




\bibitem[\protect\citeauthoryear{Brendle}{2006}]{bren06}
Brendle, S. (2006). Portfolio selection under incomplete information. Stochastic processes and their Applications, 116(5), 701-723.

\bibitem[\protect\citeauthoryear{Carpenter}{2000}]{car00}
Carpenter, J. N. (2000). Does option compensation increase managerial risk appetite?. The journal of finance, 55(5), 2311-2331.

\bibitem[\protect\citeauthoryear{Chen et al.}{2018}]{chen18}
Chen, A., Nguyen, T., Stadje, M. (2018). Risk management with multiple VaR constraints. Mathematical Methods of Operations Research, 88, 297-337. 




\bibitem[\protect\citeauthoryear{Davey and Zheng}{2022}]{davey22}
Davey, A., Zheng, H. (2022). Deep learning for constrained utility maximisation. Methodology and Computing in Applied Probability, 24(2), 661-692.


\bibitem[\protect\citeauthoryear{D{\'e}camps et al.}{2005}]{Decamps2005}D{\'e}camps, J.P., Mariotti, T. and Villeneuve, S. (2005). Investment timing under incomplete information. Mathematics of Operations Research, 30, 472--500.

\bibitem[\protect\citeauthoryear{De Franco et al.}{2019}]{fran19}
De Franco, C., Nicolle, J., Pham, H. (2019). Bayesian learning for the Markowitz portfolio selection problem. International Journal of Theoretical and Applied Finance, 22(07), 1950037.

\bibitem[\protect\citeauthoryear{Detemple}{1986}]{det86}
Detemple, J. B. (1986). Asset pricing in a production economy with incomplete information. The Journal of Finance, 41(2), 383-391. 

\bibitem[\protect\citeauthoryear{Dong and Zheng}{2020}]{dong20}
Dong, Y., Zheng, H. (2020). Optimal investment with S-shaped utility and trading and Value at Risk constraints: An application to defined contribution pension plan. European Journal of Operational Research, 281(2), 341-356.

\bibitem[\protect\citeauthoryear{Ekstr\"om and Vaicenavicius}{2016}]{eks16}
Ekstr\"om, E., Vaicenavicius, J. (2016). Optimal liquidation of an asset under drift uncertainty. SIAM Journal on Financial Mathematics, 7(1), 357-381.











\bibitem[\protect\citeauthoryear{Han and Jentzen}{2018}]{han18}
Han, J., Jentzen, A., E, W. (2018). Solving high-dimensional partial differential equations using deep learning. Proceedings of the National Academy of Sciences, 115(34), 8505-8510.



\bibitem[\protect\citeauthoryear{Ingersoll and Jin}{2013}]{ing13}
Ingersoll, J. E., Jin, L. J. (2013). Realization utility with reference-dependent preferences. The Review of Financial Studies, 26(3), 723-767.

\bibitem[\protect\citeauthoryear{Jin and Zhou}{2008}]{JZ}
Jin, H., Zhou, X.Y.(2008). Behavioral portfolio selection in continuous time. Mathematical Finance, 18(3), 385-426.

\bibitem[\protect\citeauthoryear{Kahneman and Tversky}{1979}]{KT}
Kahneman, D., Tversky, A. (1979). Prospect theory: An analysis of decision under risk. Econometrica, 47(2), 363-391.

\bibitem[\protect\citeauthoryear{Karatzas and Xue}{1991}]{kar91}
Karatzas, I., Xue, X. X. (1991). A Note On Utility Maximization Under Partial Observations 1. Mathematical Finance, 1(2), 57-70.





\bibitem[\protect\citeauthoryear{Pham}{2009}]{pham09}
Pham, H. (2009). Continuous-time stochastic control and optimization with financial applications (Vol. 61): Springer Science\& Business Media.

\bibitem[\protect\citeauthoryear{Raissi et al.}{2019}]{raissi19} 
Raissi M.,Perdikaris P. ,Karniadakis G. E. . (2019). Physics-informed neural networks: A deep learning framework
for solving forward and inverse problems involving nonlinear partial differential equations, Journal of
Computational physics, 378, 686-707.

\bibitem[\protect\citeauthoryear{Reichlin}{2013}]{rei13}
Reichlin, C. (2013). Utility maximization with a given pricing measure when the utility is not necessarily concave. Mathematics and Financial Economics, 7(4), 531-556.

\bibitem[\protect\citeauthoryear{Rieder and Bäuerle}{2005}]{rie05} 
Rieder, U., Bäuerle, N. (2005). Portfolio optimization with unobservable Markov-modulated drift process. Journal of Applied Probability, 42(2), 362-378.

\bibitem[\protect\citeauthoryear{Sass}{2007}]{sas07}
Sass, J. (2007). Utility maximization with convex constraints and partial information. Acta Applicandae Mathematicae, 97(1-3), 221-238.

\bibitem[\protect\citeauthoryear{Shin et al.}{2020}]{shin2020convergence}
Shin Y., Darbon J. ,Karniadakis G. E. (2020). On the convergence of physics informed neural networks for linear second-order elliptic and parabolic type pdes. arXiv preprint arXiv:2004.01806.


\bibitem[\protect\citeauthoryear{Wang et al.}{2022}]{wang22}
Wang, C., Li, S., He, D., Wang, L. (2022). Is $ L^ 2$ Physics informed loss always suitable for training physics informed neural network?. Advances in Neural Information Processing Systems, 35, 8278-8290.



\bibitem[\protect\citeauthoryear{Xing et al.}{2025}]{xing25}
Xing, J., Ma, J., Zheng, H. (2025).  A simple integral equation approach for  optimal investment stopping problems with partial information,  Mathematics of Operations Research, published online. 



\bibitem[\protect\citeauthoryear{Yiu}{2004}]{yiu04}
Yiu, K. F. C. (2004). Optimal portfolios under a value-at-risk constraint. Journal of Economic Dynamics and Control, 28(7), 1317-1334.



\end{thebibliography}
\end{document}